\documentclass[journal]{IEEEtran}
\usepackage{amsfonts}
\usepackage{mathrsfs}
\usepackage{amssymb,amsmath}
\usepackage[noblocks]{authblk}
\usepackage[compress]{cite}
\usepackage{bm}
\usepackage{epsfig,epstopdf, graphics,subfigure}
\usepackage{amsthm}
\usepackage{array,color}
\usepackage[compress]{cite}
\usepackage{algorithm}
\usepackage{algorithmic}
\usepackage{mathrsfs}
\usepackage{cite}
\usepackage{stfloats}
\usepackage{balance}
\usepackage{multirow}
\usepackage{tabularx}

\usepackage{multicol}
\usepackage{booktabs}

\newtheorem{proposition}{Proposition}
\newtheorem{lemma}{Lemma}

\begin{document}
	\title{Hybrid Transceiver Designs via Majorization-Minimization Algorithm Over MIMO
		Interference Channels}
	\author{Shiqi Gong, Chengwen Xing, Vincent K. N. Lau,~\IEEEmembership{Fellow,~IEEE},
		Sheng Chen,~\IEEEmembership{Fellow,~IEEE}, and Lajos Hanzo,~\IEEEmembership{Fellow, IEEE} %
		\vspace*{-10.2mm}
	}
	\maketitle
	
	\begin{abstract}
		The potential of deploying large-scale antenna arrays in future wireless systems has
		stimulated extensive research on hybrid transceiver designs aiming to approximate the
		optimal fully-digital schemes with much reduced hardware cost and signal processing
		complexity. Generally, this hybrid transceiver structure requires a joint design of
		analog and digital processing to enable both beamsteering and spatial multiplexing
		gains. In this paper, we develop various weighted mean-square-error minimization
		(WMMSE) based hybrid transceiver designs over multiple-input multiple-output (MIMO)
		interference channels at both millimeter wave (mmWave) and microwave frequencies.
		Firstly, a heuristic joint design of hybrid precoder and combiner using alternating
		optimization is proposed, in which the majorization-minimization (MM) method is
		utilized to design the analog precoder and combiner with unit-modulus constraints.
		It is validated that this scheme achieves the comparable performance to the WMMSE
		fully-digital solution. To further reduce the complexity, a phase projection based
		two-stage scheme is proposed to decouple the designs of analog and digital   { precoder-combiner}. Secondly, inspired by the fully-digital solutions based on the
		block-diagonalization zero-forcing (BD-ZF) and signal-to-leakage-plus-noise ratio
		(SLNR) criteria, low-complexity MM-based BD-ZF and SLNR hybrid designs are proposed
		to well approximate the corresponding fully-digital solutions. Thirdly, the
		partially-connected hybrid structure for reducing system hardware cost and power
		consumption is considered, for which the MM-based alternating optimization still
		works. Numerical results demonstrate the similar or superior performance of all the
		above proposed schemes over the existing benchmarks.
	\end{abstract}
	%
	
	\section{Introduction}\label{S1}
	
	The large-scale antenna array offers a promising technology in future wireless systems
	to provide ultra high data rate for bandwidth-hungry applications and the large degree
	of freedoms (DoFs) for eliminating the random effect of wireless fading channels
	\cite{lu2014overview,rusek2012scaling}. However, the
	hardware cost and implementation complexity of deploying a large number of antenna
	elements by the traditional digital signal processing are huge, because each antenna
	requires a dedicated radio frequency (RF) chain
	\cite{roh2014millimeter,hoydis2013massive}. As an alternative cost-effective solution,
	the hybrid transceiver structure with much fewer RF chains than the number of antennas
	has attracted  extensive attention recently, of which the signal processing chain
	consists of the high-dimensional analog RF precoding/combining for providing the
	beamsteering gain, followed by the low-dimensional digital baseband precoding/combing
	mainly for reaping spatial multiplexing gain
	\cite{molisch2017hybrid,el2014spatially}.
	
	For the hybrid transceiver structure, the analog RF processing can be implemented
	using phase shifters \cite{rebeiz2002rf}, switches \cite{mendez2016hybrid} and/or
	lens \cite{zeng2014electromagnetic}, among which the phase shifter based analog
	precoding/combining has been widely investigated
	\cite{liang2014low,sohrabi2016hybrid,sohrabi2017hybrid,singh2015feasibility,Liu_etal2018,ni2016hybrid,Zhou_etal2018,Xing2019hybrid}.
	Phase shifters can be used to steer transmit and receive beams towards the
	desired direction by adjusting the phase of RF signals, and thus typically
	impose constant-modulus constraints on analog precoder and combiner, which makes
	hybrid transceiver designs more complicated and challenging. It has been revealed
	that once the number of RF chains reaches twice that of data streams, implying that
	the number of phase shifters is doubled, the hybrid structure can perfectly realize
	the optimal fully-digital structure  \cite{sohrabi2016hybrid}. However, the
	application with abundant phase shifters is also impractical due to high hardware cost
	and power consumption. To alleviate this issue, the partially-connected hybrid structure
	has been proposed for enabling energy-efficient communications at the expense of
	some performance loss compared to the fully-digital structure
	\cite{sohrabi2017hybrid,singh2015feasibility,Liu_etal2018}.
	
	Hybrid transceivers are applicable not only to mmWave communications but also in
	other lower frequency range \cite{liang2014low,ni2016hybrid}. Moreover, the
	criteria of hybrid designs are diverse, e.g.,  mean squared error (MSE), capacity
	and bit error rate (BER). Various hybrid transceiver designs have been conceived
	for point-to-point MIMO systems 
	\cite{el2014spatially,Xing2019hybrid,yu2016alternating,ni2017near,chen2015iterative}
	and multiuser MIMO systems
	\cite{ni2016hybrid,kim2015mse,nguyen2017hybrid,rajashekar2017iterative,alkhateeb2015limited,wu2018hybrid}.
	The motivation of these designs is to leverage the underlying hybrid structure to
	achieve the comparable performance to the optimal (near-optimal) fully-digital
	solution. To this end, existing hybrid  designs are mainly classified into two
	categories.
	
	One category jointly designs hybrid precoder and combiner to approach the
	fully-digital performance.  For example, by exploiting the sparsity of mmWave
	channels, the orthogonal matching pursuit (OMP) algorithm  was used to jointly
	design hybrid precoder and combiner to approximate the optimal fully-digital
	solution \cite{el2014spatially}. Using matrix-monotonic optimization
	\cite{Xing2019hybrid}, the optimal unconstrained structures of analog precoder
	and combiner under various design criteria can be proved to be unitary matching
	with channel. Some heuristic joint hybrid transceiver designs via alternating
	optimization were also investigated
	\cite{yu2016alternating,ni2017near,chen2015iterative}. Specifically, to
	approximate the optimal fully-digital solution, an alternating minimization
	method was proposed for hybrid designs based on manifold optimization
	\cite{yu2016alternating} and local approximation of phase increment
	\cite{ni2017near}, respectively. In addition, joint hybrid designs were studied
	in multiuser scenarios using the minimum MSE (MMSE), WMMSE and BD-ZF
	fully-digital solutions \cite{kim2015mse,nguyen2017hybrid,rajashekar2017iterative}.
	For example, in
	\cite{nguyen2017hybrid} and \cite{rajashekar2017iterative}, the OMP algorithm
	was utilized to jointly construct the hybrid WMMSE precoder and combiner for
	achieving the performance close to the WMMSE and BD-ZF fully-digital solutions,
	respectively. However, such approaches generally require the fully-digital
	precoder to have a closed-form solution, and its applicability in more general
	scenarios may be limited.

	The other category is the two-stage hybrid transceiver design widely used in
	multiuser MIMO scenarios. In this scheme, the analog precoder and combiner are
	firstly designed by directly optimizing some performance criterion, such as the
	effective array gain. Then the digital precoder and combiner are optimized to
	further improve system performance by eliminating inter-user inference
	\cite{ni2016hybrid,alkhateeb2015limited,wu2018hybrid}. For example, in
	\cite{ni2016hybrid}, the equal gain transmission (EGT) based analog precoder
	and the discrete Fourier transform (DFT) codebook based analog combiner for
	each user were proposed to achieve large array gain. To achieve low channel
	training and feedback overhead, the two-stage hybrid design \cite{alkhateeb2015limited}
	chooses each user's analog precoder and combiner from the quantized codebooks to
	maximize effective channel gain. All the above analog processing schemes can be
	combined with the low-complexity BD-ZF digital processing \cite{spencer2004zero}
	to cancel inter-user interference. Although this BD-ZF scheme is easy to implement,
	it does not consider the influence of noise in the digital precoder design and
	thus performs poorly at low signal-to-noise ratio (SNR) regime. This fact
	motivates us to consider an effective digital processing based on the SLNR
	criterion of \cite{sadek2007leakage}. The SLNR-maximization digital processing is
	more desirable than the BD-ZF criterion in some scenarios with fewer DoFs, i.e.,
	MIMO interference channels \cite{cheng2010new}. This two-stage scheme can also be
	extended to the mixed timescale hybrid precoder optimization
	\cite{liu2014phase,liu2015two} in which the analog and digital precoders are
	adaptive to channel statistics  and real-time channel state information (CSI),
	respectively. 
	
	In this paper, we consider challenging  MIMO interference channels with very
	few DoFs and develop various hybrid transceiver designs based on the MM method.
	Since the MM method guarantees stationary convergence and has the desired
	closed-form solution of each subproblem, it offers an effective tool to address
	the nonconvex constant-modulus constraints on analog precoder and combiner
	\cite{sun2017majorization,wu2018transmit}. Specifically, we propose the
	MM-based alternating optimization, decoupled two-stage scheme and various
	low-complexity schemes for hybrid transceiver designs in both mmWave and
	lower-frequency Rayleigh MIMO interference channels. Additionally, perfect
	CSI and analog processing with infinite resolution are utilized to provide a
	theoretical performance upper-bound for practical implementation of all the
	proposed schemes.  Our contributions together with the associated
	technical challenges are summarized as follows.
	\begin{enumerate}
		\item \textbf{Joint hybrid transceiver design bypassing the optimal fully-digital
			Solution}. For the $K$-user MIMO interference channel, the joint hybrid WMMSE
		transceiver design bypassing the near-optimal fully-digital WMMSE solution is
		studied. This joint design is very challenging since the coupled variables and
		unit-modulus constraint on the analog precoder and combiner lead to the
		nonconvex and NP-hard optimization. To tackle this challenge, the MM-based
		alternating optimization under a practical property of large-scale MIMO is
		proposed, which guarantees to converge. To further reduce the computational
		complexity, we also study another phase projection (PP) based two-stage scheme
		with the decoupled designs of analog and digital precoder and combiner.
		\item \textbf{Low-complexity separate hybrid transceiver designs}. Since the
		suboptimal closed-form fully-digital precoders for each transmit-receive pair
		can be obtained based on BD-ZF and SLNR maximization (SLNR-Max) criteria,
		the proposed low-complexity hybrid transceiver designs focus on approximating
		the BD-ZF and SLNR-Max fully-digital precoders, which also belong to nonconvex
		optimization. In fact, both these low-complexity designs contain multiple
		separate hybrid transceiver designs for all transmit-receive pairs, each of
		which consists of two separate stages. To address this non-convexity, the
		iterative PP (iterative-PP) based hybrid precoder is firstly designed. Then the
		corresponding  hybrid MMSE combiner is optimized through the MM-based alternating
		optimization.
		\item \textbf{Low-cost joint hybrid transceiver design}. In order to further
		reduce hardware cost and power consumption, we consider the
		partially-connected hybrid structure, in which each RF chain  at
		transmitter/receiver is connected to a single non-overlapped subarray. In this
		context, the MM-based alternating optimization still works and can converge to the
		stationary solutions for the joint hybrid WMMSE transceiver design.
	\end{enumerate}
	\textbf{Notations}: The bold-faced lower-case and upper-case letters stand for
	vectors and matrices, respectively. The transpose, conjugate, Hermitian and
	inverse operators are denoted by $(\cdot )^{\rm T}$, $(\cdot )^*$,
	$(\cdot )^{\rm H}$ and $(\cdot )^{-1}$, respectively, while $\text{Tr}(\bm{A})$
	and $\det (\bm{A})$ denote the trace and determinant of $\bm{A}$, respectively.
	$\bm{I}_n$, $\bm{0}_{n\times m}$ and $\bm{1}_n$ are the $n\!\times\! n$ identity
	matrix, the $n\!\times\! m$ zero matrix and the $n$-dimensional vector with all
	elements being one, respectively. {The block-diagonal  matrix with  diagonal
		elements $\bm{A}_1,\cdots ,\bm{A}_N$ is denoted by $\text{BLKdiag}[\bm{A}_1,\cdots ,\bm{A}_N]$. Particularly, it is reduced to  $\text{diag}[a_1,\cdots ,a_N]$ when scalar diagonal elements  are considered  .}
	$[{\bm{A}}]_{n,m}$ denotes the $(n,m)$th {(the $n$th row and $m$th column)} element  of ${\bm{A}}$, and
	$\bm{A}(q_1:q_2,l_1:l_2)$ denotes the { sub-matrix} consisting of the $q_1$ to $q_2$
	rows and $l_1$ to $l_2$ columns of $\bm{A}$, while $\bm{A}(:,l_1:l_2)$ is the
	{ sub-matrix} consisting of the $l_1$ to $l_2$ columns of $\bm{A}$. The $n$th element
	of $\bm{a}$ is denoted by $[\bm{a}]_n$, and $\bm{a}(n:m)$ is the sub-vector
	consists of the $n$th to $m$th elements of $\bm{a}$. $\bm{A}\succ \bm{0}$
	($\succeq \bm{0}$) means that $\bm{A}$ is positive definite (semi-definite), and
	$\lambda_{\max}(\bm{A})$ is the maximum eigenvalue of $\bm{A}$, while
	$e^{\textsf{j}\arg (\cdot )}$ denotes the phase extraction operation in an
	element-wise manner. The rank of $\bm{A}$ is denoted by $\text{rank}(\bm{A })$.
	The modulus operator  denoted by $|\cdot |$, $\|\cdot \|$ is the Euclidean
	distance, and $\| \cdot \|_F$ is the matrix Frobenius norm, while
	$\mathbb{E}[\cdot ]$ is the expectation operator and $\text{vec}(\cdot )$ is the
	vectorization of a matrix. $\Re\{\cdot\}$ is the real part operator and $\otimes$
	is the Kronecker product operator, while $[a]^+\!=\!\max\{a,0\}$. The words
	`independent and identically distributed' and `with respect to' are abbreviated as
	`i.i.d.' and `w.r.t.', respectively.
	\vspace{-3mm}
	\section{System Model}\label{S2}
	
	\subsection{$K$-user MIMO interference channel}\label{S2.1}
	
	As shown in Fig.~\ref{figure0}, we consider a $K$-user MIMO interference channel, where
	all $K$ transmitters and receivers are equipped with hybrid MIMO processor for dealing
	with multiple data streams. Specifically, the $k$th transmitter equipped with $N_{t_k}$
	antennas and $N_{t_k}^{RF}$ RF chains sends $N_{s_k}$ data streams to the corresponding
	receiver equipped with $N_{r_k}$ antennas and $N_{r_k}^{RF}$ RF chains, where
	$N_{s_k}\! \le\! N_{t_k}^{RF}\! \le\! N_{t_k}$ and $N_{s_k}\! \le\! N_{r_k}^{RF}\! \le
	\! N_{r_k}$, $\forall k$. The hybrid MIMO processor at the $k$th transmitter enables 
	the digital baseband precoder $\bm{F}_{D_k}\! \in\! \mathbb{C}^{N_{t_k}^{RF}\times N_{s_k}}$,
	followed by the analog precoder $\bm{F}_{A_k}\! \in\! \mathbb{C}^{N_{t_k}\times N_{t_k}^{RF}}$.
	Similarly, the hybrid MIMO processor at the $k$th receiver consists of an analog RF
	combiner $\bm{G}_{A_k}\! \in\! \mathbb{C}^{ N_{r_k}\times N_{r_k}^{RF}}$, followed by a
	digital combiner $\bm{G}_{D_k}\! \in\! \mathbb{C}^{ N_{r_k}^{RF}\times N_{s_k}}$. Both
	$\bm{F}_{A_k}$ and $\bm{G}_{A_k}$ are realized using analog phase shifters with constant
	modulus, i.e., $\vert[\bm{F}_{A_k}]_{n,m}\vert\! =\! 1$ and $\vert[\bm{G}_{A_k}]_{n,m}
	\vert \! =\! 1$, $\forall n,m$. The transmitted signal by the $k$th transmitter is given
	by $\bm{x}_k\! =\! \bm{F}_{A_k}\bm{F}_{D_k}\bm{s_k}$, where $\bm{s}_{k}\! \in\!
	\mathbb{C}^{N_{s_k}}$ denotes the Gaussian encoded information symbols satisfying
	$\mathbb{E}[\bm{s}_k\bm{s}_k^{\rm H}]\! =\! \bm{I}_{N_{s_k}}$ and $\Vert\bm{F}_{A_k}
	\bm{F}_{D_k}\Vert_F^2\! \le\! P_k$ with $P_k$ being the maximum transmit power.
	Under the assumption of quasi-static block-fading MIMO channel, the received signal
	at the $k$th receiver is written as
	\begin{align}\label{eq1}
	\bm{y}_k =& \bm{H}_{k,k}\bm{F}_{A_k}\bm{F}_{D_k}\bm{s}_k + \sum\nolimits_{i\neq k}
	\bm{H}_{k,i}\bm{F}_{A_i}\bm{F}_{D_i}\bm{s}_i + \bm{n}_k ,
	\end{align}
	where $\bm{H}_{k,i}\! \in\! \mathbb{C}^{N_{r_k}\times N_{t_i}}$ denotes the wireless
	channel between the $i$th transmitter and $k$th receiver, and $\bm{n}_k\! \sim\!
	\mathcal{CN}(\bm{0},\sigma_{n_k}^2\bm{I}_{N_{r_k}})$ is the additive Gaussian noise
	at the $k$th receiver, which has zero mean vector and covariance matrix 
	$\sigma_{n_k}^2\bm{I}_{N_{r_k}}$. Then the hybrid analog-digital combiner at the
	$k$th receiver, i.e., $\bm{G}_k^{\rm H}\!=\! \bm{G}_{A_k}^{\rm H}\bm{G}_{D_k}^{\rm H}$, 
	is applied to $\bm{y}_k$ to obtain the desired output as
	\begin{align}\label{eq2}
	\widehat{\bm{s}}_k &= \bm{G}_{D_k}^{\rm H}\bm{G}_{A_k}^{\rm H}\bm{H}_{k,k}\bm{F}_{A_k}
	\bm{F}_{D_k}\bm{s}_k \nonumber\\
	&+ \bm{G}_{D_k}^{\rm H}\bm{G}_{A_k}^{\rm H}\sum\nolimits_{i\neq k}
	\bm{H}_{k,i}\bm{F}_{A_i}\bm{F}_{D_i}\bm{s}_i + \bm{G}_{D_k}^{\rm H}\bm{G}_{A_k}^{\rm H}
	\bm{n}_k .
	\end{align}
	\begin{figure*}[tp!]	
		\vspace{-10mm}
		\begin{center}
			\includegraphics[width=1.8\columnwidth]{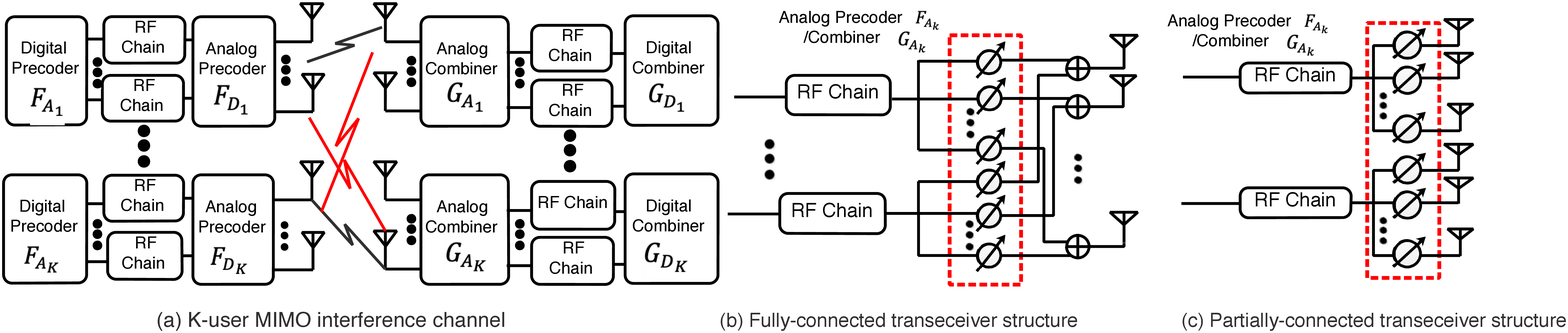}
		\end{center}
		\vspace{-5mm}
		\caption{A $K$-user MIMO interference channel with different  hybrid transceiver structures.}
		\label{figure0} 
		\hrulefill
		\vspace{-5mm}
	\end{figure*}
	
	The achievable sum rate of this $K$-user MIMO system under Gaussian signaling is
	given by
	\begin{align}\label{eq3}
	R_{\rm{sum}} &= \sum\nolimits_{k=1}^{K}\log\det(\bm{I}_{N_{s_k}} + \bm{G}_{D_k}^{\rm H}
	\bm{G}_{A_k}^{\rm H}\bm{H}_{k,k}\bm{F}_{A_k}\bm{F}_{D_k}\nonumber\\
	&~~~~~~~~~~~~~~~~ \cdot(\bm{G}_{D_k}^{\rm H}
	\bm{G}_{A_k}^{\rm H}\bm{H}_{k,k}\bm{F}_{A_k}\bm{F}_{D_k})^{\rm H}\bm{R}_k^{-1}) ,
	\end{align}
	where $\bm{R}_k\! =\! \bm{G}_{D_k}^{\rm H}\bm{G}_{A_k}^{\rm H}\big(\sum\nolimits_{i\neq k}
	\bm{H}_{k,i}\bm{F}_{A_i}\bm{F}_{D_i}\bm{F}_{D_i}^{\rm H}\bm{F}_{A_i}^{\rm H}
	\bm{H}_{k,i}^{\rm H}\! +\! \sigma_{n_k}^2\bm{I}_{N_{r_k}}\big)\bm{G}_{A_k}\bm{G}_{D_k}$
	is the covariance matrix of the inter-user interference plus noise at the $k$th receiver,
	$\forall k$. We aim to jointly design the hybrid precoders and combiners $\mathcal{A}\!
	=\! \{\bm{F}_{A_k},\bm{F}_{D_k},\bm{G}_{D_k},\bm{G}_{A_k},\forall k\}\! =\! \{
	\mathcal{A}_k,\forall k\}$ to maximize the achievable sum rate $R_{\rm{sum}}$
	\eqref{eq3}, which is formulated as
	\begin{align}\label{eq4}
	\begin{array}{cl}
	\max\limits_{\mathcal{A}_k,\forall k} & R_{\rm{sum}} , \\
	\text{s.t.} &\text{Tr}(\bm{F}_{A_k}\bm{F}_{D_k}\bm{F}_{D_k}^{\rm H}\bm{F}_{A_k}^{\rm H})\le P_k, \\
	&\vert[\bm{F}_{A_k}]_{n,m}\vert^2=1,\vert[\bm{G}_{A_k}]_{n,m}\vert^2=1,\forall k,n,m .
	\end{array}
	\end{align}
	
	Clearly, the sum rate maximization \eqref{eq4} is nonconvex and NP-hard w.r.t.
	$\mathcal{A}$ due to the coupled optimization variables and unit-modulus constraints. {
		Even the optimal fully-digital solution  $\bm{F}_{k}\!\!=\!\!\bm{F}_{A_k}\bm{F}_{D_k}$ to the problem \eqref{eq4} without the unit-modulus constraint}   has not been
	globally addressed yet, and only stationary solution generated from iterative
	process is available \cite{shi2011iteratively}. Therefore, for $K$-user MIMO
	interference channels, the traditional method of minimizing the Euclidean distance
	between the hybrid analog-digital precoder and the optimal fully-digital one
	cannot theoretically guarantee its sum rate performance. In the sequel, using a
	reasonable assumption on the analog precoder in large-scale MIMO systems, we
	find an effective joint design of hybrid precoder and combiner to the problem
	\eqref{eq4} via the MM-based alternating optimization with guaranteed sum rate
	performance.
	
	Although the proposed alternating optimization procedure achieves the semi
	closed-form solution to each subproblem, it imposes heavy coordination among
	all transmit-receive pairs. To further reduce complexity, a two-stage hybrid
	design is firstly proposed with the decoupled optimization of analog and
	digital { precoder-combiner}    for each transmit-receive pair. Then, two 
	hybrid designs based on the BD-ZF and SLNR-Max fully-digital precoding are
	studied, both support the independent hybrid precoder and combiner design.
	All the above schemes require global CSI at transmitter, which imposes huge
	training and feedback overhead. To alleviate this problem, we also consider
	the partially connected hybrid transceiver structure with significantly
	reduced feedback overhead and hardware cost, to  which the proposed alternating 
	optimization is directly applicable and a stationary solution of the problem
	\eqref{eq4} can be achieved.
	\vspace{-2mm}
	\subsection{Channel model}\label{S2.2}
	
	In our work, two kinds of block-fading channels are adopted, mmWave channels
	and Rayleigh channels. The first type considers the propagation environment
	at the mmWave band, which has limited scattering and suffers from several
	blockage and reduced diffraction, while the other considers the propagation
	environment with rich scatterers. Moreover, to make the system capacity
	independent of the scaling of the channel matrix, we use the normalized
	channel matrix.
	
	For Rayleigh channels, the elements of the channel matrix $\bm{H}_{k,i}$ are
	i.i.d. complex Gaussian variables with zero mean and unit variance, i.e.,
	$\text{vec}(\bm{H}_{k,i})\! \sim\! \mathcal{CN}(\bm{0},\bm{I}_{N_{r_k}N_{t_i}})$,
	$\forall k$ and $i\! =\! 1,\cdots, K$. For mmWave channels, the extended
	Salen-Valenzuela geometric model \cite{wallace2002modeling} is adopted:
	\begin{align}\label{eq5}
	\bm{H}_{k,i}& \!\!=  \!\!\sqrt{\frac{N_{r_k} N_{t_i}}{L_{k,i}}} \!\!\sum\nolimits_{l=1}^{L_{k,i}}
	\alpha_k^l \bm{a}_r(\theta_{k}^l)\bm{a}_t^{\rm H}(\psi_{i}^l),~\forall k , i \!\!= \!\!1,\cdots, K ,
	\end{align}
	where $L_{k,i}$ denotes the number of dominated propagation paths in the
	channel $\bm{H}_{k,i} $ and $\alpha_k^l$ is the complex gain of the $l$th path,
	while $\theta_k^l$ and $\psi_i^l$  are the angle of arrival (AOA) and angle
	of departure (AOD) of the $l$th path, respectively. Assume that the uniform
	linear array (ULA) is deployed at each transmit-receive pair. The transmit
	and receive array steering vectors can then be expressed as $\bm{a}_t(\psi_i^l)
	\! =\!\frac{1}{\sqrt{N_{t_i}}}\big[1 ~ e^{-\textsf{j}\frac{2\pi}{\lambda}\sin\psi_i^l}
	\cdots e^{-\textsf{j}(N_{t_i}-1)\frac{2\pi}{\lambda}\sin\psi_i^l}\big]^{\rm T}$
	and $\bm{a}_r(\theta_k^l)\! =\! \frac{1}{\sqrt{N_{r_k}}}\big[1 ~
	e^{-\textsf{j}\frac{2\pi}{\lambda}\sin\theta_k^l} \cdots 
	e^{-\textsf{j}(N_{r_k}-1)\frac{2\pi}{\lambda}\sin\theta_k^l}\big]^{\rm T}$,
	respectively, where $\lambda$ denotes the signal wavelength and the antenna
	element spacing is $\frac{\lambda}{2}$.
	\vspace{-1mm}
	\section{MM-based Joint Hybrid Transceiver Design}\label{S3}
	
	\subsection{Equivalent problem reformulation}\label{S3.1}
	To tackle the sum rate maximization \eqref{eq4} effectively, we introduce
	$\widetilde{\bm{F}}_{D_k}\! =\! (\bm{F}_{A_k}^{\rm H}
	\bm{F}_{A_k})^{\frac{1}{2}}\bm{F}_{D_k}$ and $\widetilde{\bm{F}}_{A_k}\! =\!
	\bm{F}_{A_k}(\bm{F}_{A_k}^{\rm H}\bm{F}_{A_k})^{-\frac{1}{2}}$, $\forall k$,
	and reformulate it as an equivalent WMMSE problem \cite{shi2011iteratively}:
	\begin{align}\label{eq6}
	\begin{array}{cl}
	\min\limits_{{\mathcal{A}_k,\bm{W}_k\succ\bm{0}}} \!\!&\!\!\!\! \sum\nolimits_{k=\!1}^K\!\!
	\big(\text{Tr}\big(\bm{W}_k \bm{E}_k({\mathcal{A}}_k)\!\big) \!\!-\!\!\log \det(\bm{W}_k) \!\!-\!\! N_{s_k}\big) , \\
	\text{s.t.}&\!\!\!\!\!\!\!\!\!\text{Tr}\big(\widetilde{\bm{F}}_{D_k}^{\rm H}\widetilde{\bm{F}}_{D_k}\big)
	\le P_k, ~ \widetilde{\bm{F}}_{A_k} \!\!=\!\! \bm{F}_{A_k}\big(\bm{F}_{A_k}^{\rm H}
	\bm{F}_{A_k}\big)^{-\frac{1}{2}}, \\
	& \!\!\!\!\!\!\!\!\!\!\big\vert[{\bm{F}}_{A_k}]_{n,m}\big\vert=1, ~ \big\vert[\bm{G}_{A_k}]_{n,m}\big\vert=1,
	~ \forall k,n,m,
	\end{array}
	\end{align}
	where we still use $\mathcal{A}_k\! =\!\{\bm{G}_{A_k},\bm{G}_{D_k},\widetilde{\bm{F}}_{A_k},
	\widetilde{\bm{F}}_{D_k}\}$ and the MSE matrix $\bm{E}_k(\mathcal{A}_k)$ is defined as
	\begin{align}\label{eq7}
	&\bm{E}_k(\mathcal{A}_k) \!= \! \mathbb{E}\big[\big(\widehat{\bm{s}}_k\! - \!\bm{s}_k\big)
	\big(\widehat{\bm{s}}_k\! \!-\!\! \bm{s}_k\big)^{\!\rm H}\big]\! \nonumber\\
	&\!\!\!= \! \!
	\big(\! \bm{G}_{D_k}\bm{G}_{A_k}\bm{H}_{k,k}\widetilde{\bm{F}}_{A_k}\widetilde{\bm{F}}_{D_k}
	\!\! -\! \!\bm{I}_{N_{s_k}}\! \big)\!\big(\! \bm{G}_{D_k}\bm{G}_{A_k}\bm{H}_{k,k}
	\widetilde{\bm{F}}_{A_k}\widetilde{\bm{F}}_{D_k}\!\!\! -\! \!\bm{I}_{N_{s_k}} \!\big)^{\!\!\rm H} \nonumber \\
	&\!\!\!+ \!\sum\limits_{i\neq k} \big(\bm{G}_{D_k}\bm{G}_{A_k}\bm{H}_{k,i}\widetilde{\bm{F}}_{A_i}
	\widetilde{\bm{F}}_{D_i}\big)\big(\bm{G}_{D_k}\bm{G}_{A_k}\bm{H}_{k,i}\widetilde{\bm{F}}_{A_i}
	\widetilde{\bm{F}}_{D_i}\big)^{\rm H} \nonumber \\
	&\!\!\!+ \sigma_{n_k}^2\bm{G}_{D_k}\bm{G}_{A_k} 
	\bm{G}_{A_k}^{\rm H}\bm{G}_{D_k}^{\rm H} .
	\end{align}
	
	For massive MIMO, the analog precoder
	design for approximating the near-optimal system performance typically satisfies
	$\bm{F}_{A_k}^{\rm H}\bm{F}_{A_k}\! \approx\! N_{t_k}\bm{I}_{N_{t_k}^{RF}}$,
	$\forall k$, with high probability when $N_{t_k}\! \rightarrow\! \infty$
	\cite{sohrabi2016hybrid,el2014spatially,wu2018hybrid}. Therefore, we exploit this
	property and assume that
	$\widetilde{\bm{F}}_{A_k}\! \approx\! \frac{1}{\sqrt{N_{t_k}}}\bm{F}_{A_k}$. Then
	the problem \eqref{eq6} is simplified as
	\begin{align}\label{eq8}
	\begin{array}{cl}
	\min\limits_{\widetilde{\mathcal{A}}_k,\bm{W}_k\succ\!\bm{0}} \!&\!\!\!\!\!\! \sum\nolimits_{k\!=\!1}^K\!
	\big(\text{Tr}\big({\bm{W}}_k \bm{E}_k(\widetilde{\mathcal{A}}_k)\big) \!\!-\!\!
	\log \det(\bm{W}_k) \!\!-\!\! N_{s_k}\big) , \\
	\text{s.t.} & \!\!\!\!\!\!\!\!\!\!\text{Tr}\big(\widetilde{\bm{F}}_{D_k}^{\rm H}\widetilde{\bm{F}}_{D_k}\big)\le P_k,~~ \big\vert[{\bm{F}}_{A_k}]_{n,m}\big\vert\!=\!1, \\
	& \!\!\!\!\!\!\!\!\!\!
	\big\vert[\bm{G}_{A_k}]_{n,m}\big\vert\!=\!1,\forall k,n,m,
	\end{array}
	\end{align}
	where $\widetilde{\mathcal{A}}_k\! =\! \{\bm{G}_{A_k},\bm{G}_{D_k},\bm{F}_{A_k},
	\widetilde{\bm{F}}_{D_k}\}$ and $\bm{E}_k\big(\widetilde{\mathcal{A}}_k\big)$ is
	obtained by using $\widetilde{\bm{F}}_{A_k}\! \approx\!
	\frac{1}{\sqrt{N_{t_k}}}\bm{F}_{A_k}$ in (\ref{eq7}). 
	
	We can jointly optimize the hybrid precoder and combiner by solving \eqref{eq8}.
	Given the other variables, the problem \eqref{eq8} is convex w.r.t.
	$\bm{W}_k$, which can be derived in closed-form
	\begin{align}\label{eqw} 
	\bm{W}_k =& \bm{E}_k^{-1}\big(\widetilde{\mathcal{A}}_k\big), ~ \forall k .
	\end{align}
	In addition, it is well-known that the optimal digital combiner $\bm{G}_{D_k}$
	for simultaneously minimizing all the MSEs of the data streams of the $k$th
	transmit-receive pair is the Wiener filter:
	\begin{align}\label{eq10}
	\bm{G}_{D_k} =& \frac{1}{\sqrt{N_{t_k}}} \widetilde{\bm{Q}}_k \bm{G}_{A_k}^{\rm H}
	\bm{H}_{k,k} \bm{F}_{A_k} \widetilde{\bm{F}}_{D_k} , ~ \forall k ,
	\end{align}
	where $\widetilde{\bm{Q}}_k\!\! =\!\!\Big(\!\sum\limits_{i=1}^K\! \frac{1}{{N_{t_i}}}\big(
	\bm{G}_{A_k}^{\rm H}\bm{H}_{k,i} \bm{F}_{A_i} \widetilde{\bm{F}}_{D_i}\big)\!
	\big(\bm{G}_{A_k}^{\rm H}\bm{H}_{k,i} \bm{F}_{A_i}
	\widetilde{\bm{F}}_{D_i}\big)^{\rm H}\! $ $+\!\sigma_{n_k}^2\bm{G}_{A_k}^{\rm H}
	\bm{G}_{A_k}\! \Big)^{-1}$. Based on the closed-form solutions \eqref{eqw} and
	\eqref{eq10}, our next task is to find the optimal solution
	$\{\bm{F}_{A_k},\bm{F}_{D_k},\bm{G}_{A_k},\forall k\}$ to the problem
	\eqref{eq8}, but \eqref{eq8} is not jointly convex
	w.r.t. $\{\bm{F}_{A_k},\bm{F}_{D_k},\bm{G}_{A_k},\forall k\}$. In Section~\ref{MMALT}, a MM-based
	alternating optimization procedure is proposed to find the semi closed-form
	solution $\{\bm{F}_{A_k},\bm{F}_{D_k},\bm{G}_{A_k},\forall k\}$  to the problem
	\eqref{eq8} with guaranteed stationary convergence.
	\vspace{-3mm}
	\section{Proposed MM-Based Alternating Optimization}\label{MMALT}
	\subsection{Brief review of MM method}\label{S4.1}
	The MM method is an effective optimization tool for solving nonconvex problems.
	The basic idea is to transform the original nonconvex problem into a sequence
	of majorized subproblems that can be solved with semi closed-from solutions and
	guaranteed convergence. The MM method generally consists of two stages, the
	majorization stage and the minimization stage. In the majorization stage, for a
	general optimization problem\vspace{-2mm}
	\begin{align}\label{eq11}
	\min\limits_{\bm{X}}~ f(\bm{X}) , ~ \text{s.t.} ~ \bm{X}\in \mathcal{X} ,
	\end{align}{ where $\mathcal{X}$ is a closed nonempty  set. In terms of our work, it can be nonconvex.} Our aim is to find a continuous surrogate function $g(\bm{X}|\bm{X}^{(l)})$, also
	defined as a majorizer of $f(\bm{X})$ at $\bm{X}^{(l)}$, for updating $\bm{X}$
	at the $l$th iteration. Mathematically, this is expressed as 
	\begin{align}\label{eq12}
	\bm{X}^{(l+1)} =&  \arg \min\limits_{\bm{X}\in \mathcal{X}}~ g(\bm{X}|\bm{X}^{(l)}) .
	\end{align}  
	
	The majorizer $g(\bm{X}|\bm{X}^{(l)})$ must satisfy the following
	conditions to ensure that the MM method converges to a stationary point of the
	problem \eqref{eq11} \cite{sun2017majorization}:
	{\begin{align}\label{eq13}
		\left\{\!\!\! \begin{array}{l}
		g(\bm{X}|\bm{X}^{(l)})\ge f(\bm{X}) , ~ \forall ~\bm{X} \in \mathcal{X} ,  \\
		g(\bm{X}^{(l)}|\bm{X}^{(l)}) = f(\bm{X}^{(l)}) , ~ \forall ~\bm{X}^{(l)}\in \mathcal{X},  \\
		g^{'}(\bm{X}^{(l)}|\bm{X}^{(l)};\bm{d}) = f^{'}(\bm{X}^{(l)};\bm{d}) , ~ \forall \bm{d} ~ \in \mathbb{T}_{\mathcal{X}}(\bm{X}^{(l)}) ,
		\end{array} \right.
		\end{align}}
	where $\mathbb{T}_{\mathcal{X}}(\bm{X}^{(l)})$ is the Boulingand tangent cone { \cite{Boulingand} } of
	$\mathcal{X}$ at $\bm{X}^{(l)}$. It is known that the limit point obtained by minimizing
	$g(\bm{X}|\bm{X}^{(l)})$ subject to $\bm{X}\! \in\! \mathcal{X}$
	satisfies the stationary condition $f^{'}(\bm{X}^{(\infty )};\bm{d})\! \ge\! 0$,
	$\forall \bm{d}\! \in\! \mathbb{T}_{\mathcal{X}}(\bm{X}^{(\infty)})$. Also, based on
	\eqref{eq13}, the monotonicity of the MM method is manifested by
	\begin{align}\label{eq14}
	& f(\bm{X}^{(l+1)})\!\le\! g(\bm{X}^{(l+1)}|\bm{X}^{(l)})\!\le\! g(\bm{X}^{(l)}|\bm{X}^{(l)}) \!=\!f(\bm{X}^{(l)}),\forall l .
	\end{align}
	The interested readers can refer to \cite{sun2017majorization,wu2018transmit} for more
	details of the general MM method.
	\vspace{-3mm}
	\subsection{Proposed MM-based alternating  optimization}\label{MMALT1}
	Our proposed MM-based alternating optimization for the problem \eqref{eq8} is a
	combination of the block coordinate descent (BCD) and MM methods. To be specific,
	we first partition the remaining variables into three blocks as
	$\{\widetilde{\bm{F}}_{D_k},\forall k\}$, $\{\bm{F}_{A_k},\forall k\}$ and
	$\{\bm{G}_{A_k},\forall k\}$. The MM method is then utilized to update the blocks
	$\{\bm{F}_{A_k},\forall k\}$ and $\{\bm{G}_{A_k},\forall k\}$, respectively, with
	the other blocks fixed. Compared to applying the MM method to the problem \eqref{eq8}
	with a single complete block, this approach provides more flexibility in designing
	surrogate functions for better approximating the objective function of the problem
	\eqref{eq8}, leading to faster convergence rate \cite{sun2017majorization}.
	
	\subsubsection{Semi closed-form digital precoder $\{\widetilde{\bm{F}}_{D_k},\forall k\}$}
	
	Given the fixed $\{\bm{F}_{A_k},\bm{G}_{A_k},\bm{G}_{D_k},\bm{W}_k,\forall k\}$, we can
	rewrite the objective function of the problem \eqref{eq8} by omitting the constant term as
	\begin{align}\label{eq15}
	F_{\rm{obj}}(\widetilde{\mathcal{A}}) &= \sum\nolimits_{k=1}^K \text{Tr}\big(\bm{W}_k
	\bm{E}_k(\widetilde{\mathcal{A}}_k) \big) \nonumber\\
	&= \sum\nolimits_{k=1}^K \sum\nolimits_{i=1}^K
	\text{Tr}\big(\widetilde{\bm{F}}_{D_k}^{\rm H}\bm{L}_{i,k}^{\rm H}\bm{W}_i
	\bm{L}_{i,k}\widetilde{\bm{F}}_{D_k}\big)  \\
	& - \sum\nolimits_{k=1}^K \text{Tr}\big(\bm{W}_k\bm{L}_{k,k}\widetilde{\bm{F}}_{D_k} +
	\bm{W}_k\widetilde{\bm{F}}_{D_k}^{\rm H}\bm{L}_{k,k}^{\rm H}\big) + C_1 ,\nonumber
	\end{align}
	where $\bm{L}_{i,k}\!\! =\!\frac{1}{\sqrt{N_{t_k}}}\bm{G}_{D_i}^{\rm H}\bm{G}_{A_i}^{\rm H}
	\bm{H}_{i,k}{\bm{F}}_{A_k}$, $\forall i,k$, and $C_1\! =\!\sum\nolimits_{k=1}^K
	\text{Tr}(\bm{W}_k\! +\! \sigma_{n_k}^2 \bm{G}_{A_k}\bm{G}_{D_k} \bm{W}_k\bm{G}_{D_k}^{\rm H}
	\bm{G}_{A_k}^{\rm H})$. Taking the derivative of $F_{\rm{obj}}(\widetilde{\mathcal{A}})$
	w.r.t. $\widetilde{\bm{F}}_{D_k}$ leads to the semi closed-form digital precoder
	\begin{align}\label{eq16}
	\widetilde{\bm{F}}_{D_k}\! \!  =\! \! \Big(\! \sum\nolimits_{i=1}^K \!  \bm{L}_{i,k}^{\rm H} \bm{W}_i
	\bm{L}_{i,k} \! +\!  \beta_k\bm{I}_{N_{t_k}^{RF}}\Big)^{-1} \bm{L}_{k,k}^{\rm H} \bm{W}_k , ~\forall k,
	\end{align}
	where $\beta_k$ is the dual variable associated with the $k$th transmit power
	constraint. Define the eigenvalue decomposition (EVD) $\sum\nolimits_{i=1}^K\bm{L}_{i,k}^{\rm H} 
	\bm{W}_i\bm{L}_{i,k}\! =\! \bm{U}_{L_k}\bm{\Lambda}_{L_k}\bm{U}_{L_k}^{\rm H}$. Since
	$\beta_k$ should satisfy the complementarity slackness condition $\beta_k\big(
	\text{Tr}(\widetilde{\bm{F}}_{D_k}\widetilde{\bm{F}}_{D_k}^{\rm H })-P_k\big)\! =\! 0$,
	if $\text{Tr}(\widetilde{\bm{F}}_{D_k}\widetilde{\bm{F}}_{D_k}^H  )\! \le \!P_k$, the
	optimal $\beta_k^{\rm{opt}}\! = \! 0$; otherwise, $\beta_k^{\rm{opt}}$ is derived from
	$\sum\nolimits_{m=1}^{N_{t_k}^{RF}} \frac{[\bm{Q}_k\bm{Q}_k^{\rm H}]_{m,m}}
	{([\bm{\Lambda}_{L_k}]_{m,m}+\beta_k^{\rm{opt}})^2}=P_k$, where $\bm{Q}_k\! =\!
	\bm{U}_{L_k}^{\rm H}\bm{L}_{k,k}^{\rm H}\bm{W}_k$.
	
	\subsubsection{Semi closed-form analog precoder $\{\bm{F}_{A_k},\forall k\}$}
	Given $\{\widetilde{\bm{F}}_{D_k},\bm{G}_{A_k},\bm{G}_{D_k},\bm{W}_{k},
	\forall k\}$, the problem \eqref{eq8} is non-convex on $\bm{F}_{A_k},\forall k$.
	We need to find an effective majorizer of the objective function of \eqref{eq8}
	in terms of $\bm{F}_{A_k}$, $\forall k$, so that a stationary solution
	$\{\bm{F}_{A_k},\forall k\}$ for the problem \eqref{eq8} can be obtained using
	the MM method. According to the identity $\text{Tr}(\bm{A}\bm{B}\bm{C}\bm{D})\! =\!
	\text{vec}(\bm{A}^{\rm T})^{\rm T}(\bm{D}^{\rm T}\otimes \bm{B})\text{vec}(\bm{C})$,
	the objective function of \eqref{eq8} can be re-expressed as
	\vspace{-1mm}
	\begin{align}\label{eq17}
	G_{\rm{obj}}(\widetilde{\mathcal{A}}) =& \sum\nolimits_{k=1}^{K} \bm{f}_{A_k}^{\rm H}
	\widetilde{\bm{A}}_k \bm{f}_{A_k}-2 \Re\big\{\bm{a}_k^{\rm H} \bm{f}_{A_k}\big\} + C_1 ,
	\end{align}
	where $\bm{f}_{A_k}\! =\! \text{vec}(\bm{F}_{A_k})$, $\widetilde{\bm{A}}_k\! =\!
	\frac{1}{N_{t_k}}\sum\nolimits_{i=1}^K\big( (\widetilde{\bm{F}}_{D_k}^{*}
	\widetilde{\bm{F}}_{D_k}^{\rm T})\! \otimes\! (\bm{M}_{i,k}\bm{W}_i
	\bm{M}_{i,k}^{\rm H})\big)$, $\bm{M}_{i,k}\! = \! \bm{H}_{i,k}^{\rm H}
	\bm{G}_{A_i}\bm{G}_{D_i}$, and $\bm{a}_k^{\rm H}\! =\! \frac{1}{\sqrt{N_{t_k}}}
	\text{vec}(\bm{W}_k^{\rm T})^{\rm T}\big(\widetilde{\bm{F}}_{D_k}^{\rm T}\!
	\otimes\! \bm{M}_{k,k}^{\rm H }\big)$, $\forall i, k$. It is clearly observed 
	that there is no coupling among $\bm{f}_{A_k},\forall k$ in
	$G_{\rm{obj}}(\widetilde{\mathcal{A}})$, implying that the designs of
	$\bm{f}_{A_k},\forall k$, are independent of each other.
	\begin{lemma}\label{Le1}\cite{wu2018transmit}
		For any two Hermitian matrices $\bm{Q},\bm{Y}\! \in\! \mathbb{C}^{N\times N}$
		satisfying $\bm{Q}\! \succeq\! \bm{Y}$, a majorizer of the quadratic function
		$\bm{x}^{\rm H}\bm{Y}\bm{x}$ at any point $\bm{x}_0\! \in\! \mathbb{C}^{N}$ is
		$\bm{x}^{\rm H}\bm{Q}\bm{x}\! +\! 2\Re(\bm{x}^{\rm H}(\bm{Y}\! -\! \bm{Q})\bm{x}_0)
		\! +\! \bm{x}_0^{\rm H}(\bm{Q}\! -\! \bm{Y})\bm{x}_0.$
	\end{lemma}
	\vspace{-1mm}
	According to Lemma~\ref{Le1}, a majorizer $g(\bm{f}_{A_k}|\bm{f}_{A_k}^{(l)})$ of
	{$G_{\rm{obj}}(\widetilde{\mathcal{A}})$} at $\bm{f}_{A_k}^{(l)}$ can be constructed as
	\begin{align}\label{eq18}
	g(\bm{f}_{A_k}|\bm{f}_{A_k}^{(l)}) &=\lambda_{\max}(\widetilde{\bm{A}}_k)
	\bm{f}_{A_k}^{\rm H} \bm{f}_{A_k} + 2\Re\{\bm{f}_{A_k}^{\rm H}\widetilde{\bm{a}}_k\}
	\nonumber\\
	& + (\bm{f}_{A_k}^{(l)})^{\rm H}(\lambda_{\max}(\widetilde{\bm{A}}_k)\bm{I}_{N_{t_k}N_{t_k}^{RF}}
	- \widetilde{\bm{A}}_k) \bm{f}_{A_k}^{(l)} ,
	\end{align}
	where $\widetilde{\bm{a}}_k\! \!=\! \!\big(\!\widetilde{\bm{A}}_k\!\! -\!\!
	\lambda_{\max}(\widetilde{\bm{A}}_k)\bm{I}_{N_{t_k}N_{t_k}^{RF}}\big) \bm{f}_{A_k}^{(l)}
	\! \!-\!\! \bm{a}_k$. Hence, the majorized problem for optimizing $\bm{f}_{A_k}$,
	can be formulated as 
	\begin{align}\label{eq19}
	& \min\limits_{\bm{f}_{A_k}}  ~\Re\big\{\bm{f}_{A_k}^{\rm H}\widetilde{\bm{a}}_k\big\},~\text{s.t.} ~\vert[\bm{f}_{A_k}]_q\vert\!=\!1 , \forall q\!=\!\!1\cdots N_{t_k}N_{t_k}^{RF}, 
	\end{align}
	{ The semi closed-form solution
		for (\ref{eq19}) is given by}
	\begin{align}\label{app} 
	\bm{f}_{A_k} =& - e^{\textsf{j}\arg\big(\widetilde{\bm{a}}_k\big)} , ~ \forall k .
	\end{align}
	
	\subsubsection{Semi closed-form analog combiner $\{\bm{G}_{A_k},\forall k\}$}
	
	Similarly, by fixing $\{\bm{G}_{D_k},\bm{F}_{A_k},\widetilde{\bm{F}}_{D_k},\bm{W}_k,
	\forall k\}$, we re-express the objective function of the problem \eqref{eq8} in
	terms of $\bm{G}_{A_k},\forall k$, as
	\begin{align}\label{eq21}
	S_{\rm{obj}}\big(\widetilde{\mathcal{A}}\big) =& \sum\nolimits_{k=1}^K
	\bm{g}_{A_k}^{\rm H} \widetilde{\bm{N}}_{k} \bm{g}_{A_k} - 2\Re\{\bm{d}_k^{\rm H}\bm{g}_{A_k}\} ,
	\end{align}	 
	where $\bm{g}_{A_k}\! =\! \text{vec}(\bm{G}_{A_k})$, $\bm{d}_k^{\rm H}\!=\!
	\text{vec}(\bm{D}_k^{\rm T}\bm{G}_{D_k}^{\rm T})^{\rm T}$, $\bm{D}_k\! =\!
	\frac{1}{\sqrt{N_{t_k}}}\bm{W}_k\widetilde{\bm{F}}_{D_k}^{\rm H}\bm{F}_{A_k}^{\rm H}
	\bm{H}_{k,k}^{\rm H}$ and $\widetilde{\bm{N}}_k\! =\! (\bm{G}_{D_k}\bm{W}_k
	\bm{G}_{D_k}^{\rm H})^{\rm T}\otimes $ $\big(\! \sum\nolimits_{i=1}^K
	\frac{1}{N_{t_k}} \bm{H}_{k,i}\bm{F}_{A_i}\widetilde{\bm{F}}_{D_i}
	\widetilde{\bm{F}}_{D_i}^{\rm H}\bm{F}_{A_i}^{\rm H}\bm{H}_{k,i}^{\rm H}\! \!+\!
	\sigma_{n_k}^2\bm{I}_{N_{r_k}}\big)$. Obviously, $S_{\rm{obj}}(\widetilde{\mathcal{A}})$
	is separable w.r.t. {$\bm{g}_{A_k}$}. Hence, based on Lemma~\ref{Le1},
	a majorizer $s(\bm{g}_{A_k}|\bm{g}_{A_k}^{(l)})$ 
	of $S_{\rm{obj}}(\widetilde{\mathcal{A}})$
	at $\bm{g}_{A_k}^{(l)}$ is given by 
	\begin{align}\label{eq22}
	s(\bm{g}_{A_k}|\bm{g}_{A_k}^{(l)}) =& \lambda_{\max}(\widetilde{\bm{N}}_k)\bm{g}_{A_k}^{\rm H}
	\bm{g}_{A_k} + 2\Re\big\{\bm{g}_{A_k}^{\rm H}\widetilde{\bm{d}}_k\big\} + C_3 ,
	\end{align}
	where $\widetilde{\bm{d}}_k\! =\! (\widetilde{\bm{N}}_{k}\! -\! \lambda_{\max}(\widetilde{\bm{N}}_{k})
	\bm{I}_{N_{r_k}N_{r_k}^{RF}})\bm{g}_{A_k}^{(l)}\! -\! \bm{d}_k$ and $C_3\! =\!
	\text{Tr}(\bm{W}_k)\! +\! (\bm{g}_{A_k}^{(l)})^{\rm H}(\lambda_{\max}(\bm{N}_{k})
	\bm{I}_{N_{r_k}N_{r_k}^{RF}}\! - \!\widetilde{\bm{N}}_k)\bm{g}_{A_k}^{(l)}$. 
	Hence, the majorized problem for optimizing $\bm{g}_{A_k}$ can be simplified as
	\begin{align}\label{eq23}
	\min\limits_{\bm{g}_{A_k}} ~\Re\big\{\bm{g}_{A_k}^{\rm H}\widetilde{\bm{d}}_k\big\},
	\text{s.t.} ~ \vert[\bm{g}_{A_k}]_{q^{'}}\vert\!\!=\!\!1 ,  ~\forall q^{'}\!\!=\!\!1\cdots N_{r_k}N_{r_k}^{RF} ,
	\end{align}  \vspace{-2mm}
	with the semi closed-form solution
	\begin{align}\label{eq24} 
	\bm{g}_{A_k} =& - e^{\textsf{j}\arg\big(\widetilde{\bm{d}}_k\big)} .
	\end{align}  
	
	Integrating the solutions \eqref{eqw}, \eqref{eq10}, \eqref{eq16}, \eqref{app} and
	\eqref{eq24} leads to the proposed MM-based alternating optimization for the
	hybrid transceiver design, which is listed in Algorithm~\ref{AL1}.
	
	\setlength{\textfloatsep}{0.3cm}
	{\renewcommand\baselinestretch{0.9}\selectfont
		\begin{algorithm}[tp]
			\caption{MM-Alt-Opt: Joint hybrid transceiver design for the problem \eqref{eq8}}
			\label{AL1}
			\begin{small}
				\begin{algorithmic}[1]
					\REQUIRE{Initial hybrid precoders and combiners $\widetilde{\mathcal{A}}^{(0)}$; outer
						iteration index $I_W=0$; convergence threshold $\epsilon_{obj}$.}
					\REPEAT
					\STATE{Calculate $\bm{G}_{D_k}^{(I_{W}+1)},\forall k$, according to \eqref{eq10}.}
					\STATE{Fix $\bm{G}_{D_k}^{(I_{W}+1)},\forall k$, calculate $\bm{W}_k^{(I_{W}+1)},
						\forall k$, according to \eqref{eqw}.}
					\STATE{Fix $\{\bm{W}_{k}^{(I_{W}+1)},\bm{G}_{D_k}^{(I_{W}+1)},\forall k\}$,
						calculate $\widetilde{\bm{F}}_{D_k}^{(I_{W}+1)},\forall k$, according to \eqref{eq16}.}
					\STATE{Update $\widetilde{\mathcal{A}}^{(I_{W})}$ with $\{\widetilde{\bm{F}}_{D_k}^{(I_{W}\!+\!1)},
						\bm{W}_{k}^{(I_{W}\!+\!1)},\bm{G}_{D_k}^{(I_{W}\!+\!1)},\forall k\}$.}
					\STATE{\text{Calculate $\widetilde{\bm{F}}_{A_k}^{(I_{W}+1)},\forall k$, using MM method}
						$\begin{cases} ~~ \textbf{Input }\text{inner iteration index } I_{M}\!=\!0 \text{ and set }
						\widetilde{\mathcal{A}}^{(I_{M})}\!=\!\widetilde{\mathcal{A}}^{(I_{W})}. \\
						~~~~ \textbf{while } \vert  F_{\rm{obj}}(\widetilde{\mathcal{A}}^{(I_{M})}) -
						F_{\rm{obj}}(\widetilde{\mathcal{A}}^{(I_{M}-1)})\vert\le \epsilon_{obj} \textbf{ do} \\
						~~~~~~ \text{Calculate } \bm{F}_{A_k}^{(I_{M}+1)},\forall k, \text{ according to \eqref{app}}. \\
						~~~~~~ \text{Set } I_M=I_M + 1 \text{ and update } \widetilde{\mathcal{A}}^{(I_{M})}
						\text{ with } \bm{F}_{A_k}^{(I_{M})},\forall k. \\
						~~~~ \textbf{end while} \\
						~~ \textbf{Output } \bm{F}_{A_k}^{(I_{W}+1)}=\bm{F}_{A_k}^{(I_{M})},\forall k.
						\end{cases}$}
					\STATE{Update $\widetilde{\mathcal{A}}^{(\!I_{W}\!)}$ with $\!\{\bm{W}_{k}^{(I_{W}\!+\!1)},
						\bm{G}_{D_k}^{(I_{W}\!+\!1)},\widetilde{\bm{F}}_{D_k}^{(I_{W}\!+\!1)},\bm{F}_{A_k}^{(I_{W}\!+\!1)},\forall k\}\!$.}
					\STATE{\text{Calculate $\bm{G}_{A_k}^{(I_W+1)},\forall k$, using MM method}
						$\begin{cases} ~~ \textbf{Input} \text{ inner iteration index } I_M=0 \text{ and set }
						\widetilde{\mathcal{A}}^{(I_M)}\!=\!\widetilde{\mathcal{A}}^{(I_{W})}. \\
						~~~~ \textbf{while } \vert  F_{\rm{obj}}(\widetilde{\mathcal{A}}^{(I_{M})}) -
						F_{\rm{obj}}(\widetilde{\mathcal{A}}^{(I_{M}-1)})\vert\le \epsilon_{obj} \textbf{ do} \\
						~~~~~~ \text{Calculate } \bm{G}_{A_k}^{(I_{M}+1)} \text{ according  to \eqref{eq24}}.\\
						~~~~~~ \text{Set } I_M=I_M +1 \text{ and update } \widetilde{\mathcal{A}}^{(I_M)}
						\text{ with } \bm{G}_{A_k}^{(I_M)},\forall k. \\
						~~~~ \textbf{end while} \\
						~~\textbf{Output } \bm{G}_{A_k}^{(I_{W}+1)}=\bm{G}_{A_k}^{(I_M)},\forall k.
						\end{cases}$}
					\STATE{Set $\widetilde{\mathcal{A}}^{(I_W+1)}\!\!=\!\!\{\bm{W}_{k}^{(I_W+1)},\bm{G}_{D_k}^{(I_W+1)},
						\widetilde{\bm{F}}_{D_k}^{(I_W+1)},\bm{F}_{A_k}^{(I_W+1)},$ $\bm{G}_{A_k}^{(I_W+1)}, \forall k\}$ and $I_W=I_W +1$.}
					\UNTIL {$\vert F_{\rm{obj}}(\widetilde{\mathcal{A}}^{(I_{M})}) -
						F_{\rm{obj}}(\widetilde{\mathcal{A}}^{(I_{M}-1)})\vert\le \epsilon_{obj}$.}
					\ENSURE{ $\{\bm{G}_{D_k}^{(I_W)},\bm{G}_{A_k}^{(I_W)},\bm{F}_{D_k}^{(I_W)},\bm{F}_{A_k}^{(I_W)},
						\forall k\}$ based on $\widetilde{\bm{F}}_{D_k}=\big(\bm{F}_{A_k}^{\rm H}
						\bm{F}_{A_k}\big)^{\frac{1}{2}}\bm{F}_{D_k}$.}
				\end{algorithmic}
			\end{small}
		\end{algorithm}}
		\setlength{\floatsep}{0.05mm}
		\vspace{-3mm}
		\subsection{Two-stage hybrid transceiver design}\label{S4.3}
		
		Next we propose a two-stage hybrid transceiver design with the decoupled
		analog and digital { precoder-combiner} optimization. First, we present a
		useful property of large-scale MIMO. 
		\begin{proposition}\label{p1}
			For large-scale MIMO systems with Rayleigh or mmWave channels, the correlation
			matrices between different channels $\bm{H}_{k,k}$ and $\bm{H}_{i,k},
			\forall i\neq k$, satisfy
			\begin{align} \label{Ho1} 
			&\lim_{N_{t_k}\to +\infty}\frac{1}{N_{t_k}}\bm{H}_{i,k}\bm{H}_{k,k}^{\rm H} \!=\!
			\bm{0}_{N_{r_i}\times N_{r_k}} , \forall i, k\!=\!1,\cdots,K , ~ i \!\neq\! k .
			\end{align}
			Moreover, define the singular value decompositions (SVDs) $\bm{H}_{i,k}\!\! =\!
			\bm{U}_{i,k}\bm{\Lambda}_{i,k}\bm{V}_{i,k}^H,\forall i,k$. We infer from
			\eqref{Ho1} that the first $\text{rank}(\bm{H}_{k,k})$ columns of $\bm{V}_{k,k}$
			and the first
			$\text{rank}(\bm{H}_{i,k})$ columns of $\bm{V}_{i,k},\forall i\neq k$, are
			asymptotically orthogonal, i.e.,
			\begin{align}\label{Ho2} 
			&\lim_{N_{t_k}\to +\infty}{\bm{V}_{i,k}^{\rm H}
				(1:\text{rank}(\bm{H}_{i,k}),:)}
			\bm{V}_{k,k}(:,1:\text{rank}(\bm{H}_{k,k})) \nonumber\\
			&~~~~=
			\bm{0}_{\text{rank}(\bm{H}_{i,k})\times \text{rank}(\bm{H}_{k,k})},~~ \forall i \neq k.
			\end{align}
		\end{proposition}
		\vspace{-1mm}
		\begin{proof}
			See Appendix~\ref{APa}.
		\end{proof}
		
		According to Proposition~\ref{p1}, the desired channel of each transmit-receive pair
		and the corresponding interference channels are all asymptotically orthogonal, which
		implies that the inter-user interference can be canceled by the large array effect
		without loss of MIMO transceiver design freedom. Based on the above discussion, in
		the first stage, we independently design the analog precoder and combiner
		$\{\bm{F}_{A_k},\bm{G}_{A_k}\}$ of each transmit-receive pair to maximize the
		effective channel gain, which is beneficial to improve the sum rate. Mathematically,
		it is formulated as\vspace{-1mm}
		\begin{align} \label{Tw1} %
		&\min\limits_{\bm{G}_{A_k}} ~~\Vert \bm{G}_{A_k}^{\rm H} \bm{H}_{k,k}\bm{F}_{A_k}
		\Vert_F^2 , \nonumber\\
		&~ \text{s.t.} ~~~ \big\vert[{\bm{F}}_{A_k}]_{n,m}\big\vert=1 ,
		\big\vert[\bm{G}_{A_k}]_{n,m}\big\vert=1,\forall n,m .
		\end{align}
		Note that the problem \eqref{Tw1} is still nonconvex. However, for large-scale MIMO
		systems, the unconstrained optimal solution of the problem \eqref{Tw1} is easily
		derived as \cite{wu2018hybrid,sadek2007leakage}
		\vspace{-1mm}
		\begin{align}\label{Tw2} 
		\bm{F}_{A_k}^{\rm{Unc}} = \bm{V}_{k,k}(:,1:N_{t_k}^{RF})  , ~
		\bm{G}_{A_k}^{\rm{Unc}} = \bm{U}_{k,k}(:,1:N_{r_k}^{RF}) .
		\end{align}
		
		Our goal is to design the unit-modulus analog precoder and combiner to
		sufficiently approximate the closed-form solution of \eqref{Tw2}. { Hence, the unit-modulus analog precoder}
		$\bm{F}_{A_k}$ is designed so that
		\begin{align}\label{eq29}
		\min\limits_{\bm{F}_{A_k}} \Vert \bm{F}_{A_k} - \bm{F}_{A_k}^{\rm{Unc}}\Vert_F^2 ,
		~ \text{s.t.} ~ \vert[\bm{F}_{A_k}]_{n,m}\vert=1, \forall n,m.
		\end{align}
		Likewise, the unit-modulus analog combiner $\bm{G}_{A_k}$ can also be obtained
		according to
		\begin{align}\label{eq30}
		\min\limits_{\bm{G}_{A_k}} \Vert \bm{G}_{A_k} - \bm{G}_{A_k}^{\rm{Unc}}\Vert_F^2 ,
		~ \text{s.t.} ~ \vert[\bm{G}_{A_k}]_{n,m}\vert=1, \forall n,m.
		\end{align}
		Both these two problems can be globally solved by PP to yield the closed-form
		solutions as
		\begin{align}\label{TWA2} 
		\bm{F}_{A_k}^{\rm{PP}} =  e^{\text{j}\arg \big(\bm{F}_{A_k}^{\rm{Unc}}\big)} , ~
		\bm{G}_{A_k}^{\rm{PP}} =  e^{\text{j}\arg \big(\bm{G}_{A_k}^{\rm{Unc}}\big)} .
		\end{align}
		
		In the second digital stage, to further suppress the inter-user interference at
		all transmit-receive pairs, the WMMSE-based joint optimization of the digital
		precoder and combiner is still required. By fixing the analog precoder and 
		combiner at the solutions obtained in the first analog stage, a low-dimensional
		alternating optimization between the digital precoder $\widetilde{\bm{F}}_{D_k}$
		in \eqref{eq16} and the digital combiner $\bm{G}_{D_k}$ in \eqref{eq10} is
		performed, which clearly has lower computational complexity than the proposed
		MM-based alternating optimization of Section~\ref{MMALT1}.
		
		This two-stage hybrid design can be regarded as a special case of the MM-based
		alternating optimization by predetermining the analog precoder and combiner of
		each transmit-receiver pair as given in \eqref{TWA2}, and thus only the
		iterative procedure between the digital precoder $\bm{F}_{D_k}$ and the
		digital combiner $\bm{G}_{D_k}$ is performed. The performance of {the MM-based alternating} optimization  generally depends on the initial point
		\cite{shi2011iteratively}, and we heuristically choose the analog precoder and
		combiner design of \eqref{TWA2} as a initial point due to its potential in
		harvesting large array gain. The superior sum rate performance of this two-stage
		hybrid design will be illustrated by the numerical simulations of Section~\ref{SIM}.
		\vspace{-3mm}
		\section{Low-Complexity and { Low-Cost} Hybrid Transceiver Designs}\label{S5}
		Although the semi closed-form solutions to hybrid transceiver design can be
		obtained using the above two alternating optimization procedures, they require 
		extensive coordination among all transmit-receive pairs. In this section, we
		investigate the low-complexity hybrid transceiver designs from the perspectives
		of decoupling hybrid precoder and combiner designs for each transmit-receive
		pair and  reducing hardware cost, respectively. 
		\vspace{-4mm}
		\subsection{BD-ZF hybrid transceiver design}\label{S5.1}
		It is well-known that by completely eliminating the inter-user interference, the 
		BD-ZF precoding is a near-optimal scheme for multiuser massive MIMO systems. By
		considering it for our MIMO interference channels, we firstly propose a
		low-complexity BD-ZF hybrid transceiver design, in which the number of antennas
		at each transmitter is larger than the total number of receive antennas, i.e.,
		$N_{t_k}\! >\!\! \sum\nolimits_{i=1}^K N_{r_i}$, $\forall k$. For mmWave channels,
		this restriction can be relaxed further to $N_{t_k}\! >\!\! \sum\nolimits_{i=1}^K
		\text{rank}(\bm{H}_{i,k})$, $\forall k$. Specifically, by first defining the 
		leakage channel for the $k$th transmit-receive pair as $\widetilde{\bm{H}}_k\! = \!
		\big[\bm{H}_{1,k}^{\rm H} \cdots \bm{H}_{k-1,k}^{\rm H} ~ \bm{H}_{k+1,k}^{\rm H}
		\cdots \bm{H}_{K,k}^{\rm H}\big]$, $\forall k$, an orthonormal basis for the
		orthogonal complement of $\widetilde{\bm{H}}_{k}$ is given by $\widetilde{\bm{H}}_k^{\bot}
		\! \in \!\mathbb{C}^{N_{t_k}\times L_k}$ with $L_k\! = \!\big(N_{t_k}\! -
		\! \sum\nolimits_{i\neq k} N_{r_i}\big)\! \ge\! N_{s_k}$ and
		$(\widetilde{\bm{H}}_{k}^{\bot})^{\rm H}\widetilde{\bm{H}}_{k}^{\bot}\! =\! \bm{I}_{L_k}$.
		Then the fully-digital BD-ZF precoder $\bm{F}_k^{\rm{ZF}}$ at the $k$th transmitter
		for eliminating  both inter-user and intra-data interference can be expressed as
		\begin{align}\label{ZF1} 
		\bm{F}_k^{\rm{ZF}} =& \widetilde{\bm{H}}_{k}^{\bot} \widetilde{\bm{V}}_k(:,1:N_{s_k})
		\sqrt{\bm{\Lambda}}_k, ~ \forall k ,
		\end{align}
		where $\widetilde{\bm{V}}_k\! \in\! \mathbb{C}^{L_k\times L_k}$ originates from the
		SVD $\bm{H}_{k,k} \widetilde{\bm{H}}_{k}^{\bot}\! = \! \widetilde{\bm{U}}_k
		\widetilde{\bm{\Lambda}}_k\widetilde{\bm{V}}_k^H$ with $\widetilde{\bm{\Lambda}}_k\!
		=\! \text{diag}\big[\widetilde{\lambda}_{k,1}^2,\cdots ,\widetilde{\lambda}_{k,L_k}^2
		\big]$, and $\bm{\Lambda}_k\! =\! \text{diag}[f_{k,1},\cdots ,f_{k,N_{s_k}}]$ is the
		solution of the following sum rate maximization
		\begin{align}\label{ZF2} 
		\begin{array}{cl}
		\max\limits_{\{f_{k,s}, \forall k,s\}}&~ \sum\nolimits_{k=1}^{K}\sum\nolimits_{s=1}^{N_{s_k}}
		\log\big(1\! +\! \sigma_{n_k}^{-2}\widetilde{\lambda}_{k,s}^2 f_{k,s}\big) ,\\ 
		{\text{s.t.}} &~ \sum\nolimits_{s=1}^{N_{s_k}} f_{k,s} \le P_k , ~\forall k .
		\end{array}
		\end{align} It is clear that the optimal solution to the problem \eqref{ZF2} has a water-filling
		structure, i.e., $f_{k,s}\! =\! \big[\frac{1}{\mu\ln2}-\frac{\sigma_{n_k}^2}
		{\widetilde{\lambda}_{k,s}^2}\big]^{+}$, $\forall k,l$, where $\mu$ is chosen to
		satisfy $\sum\nolimits_{l=1}^{N_{s_k}} f_{k,l}\! =\! P_k$, $\forall k$.
		\subsubsection{Iterative-PP hybrid precoder design}
		Given the fully-digital BD-ZF precoder \eqref{ZF1}, we design the hybrid precoder by
		solving the following optimization problem
		\begin{align}\label{ZF3} 
		&\min\limits_{\bm{F}_{A_k},\bm{F}_{D_k}} ~\Vert \bm{F}_k^{\rm{ZF}} - \bm{F}_{A_k}
		\bm{F}_{D_k}\Vert_F^2 ,\\
		& ~~~\text{s.t.} ~~~ \vert[\bm{F}_{A_k}]_{n,m}\vert=1, 
		\Vert{\bm{F}}_{A_k}\bm{F}_{D_k}\Vert_F^2=P_k , ~ \forall n,m, k,\nonumber
		\end{align} 
		where the maximum power transmission is adopted. By introducing the new variables
		$\widetilde{\bm{F}}_{D_k}\!\!=\!\! (\bm{F}_{A_k}^{\rm H} \bm{F}_{A_k})^{\frac{1}{2}}
		\bm{F}_{D_k}$ and { $\widetilde{\bm{F}}_{A_k}\!=\!$ $ \bm{F}_{A_k} (\bm{F}_{A_k}^{\rm H}
			\bm{F}_{A_k})^{-\frac{1}{2}}$,} $\forall k$, the problem \eqref{ZF3} is rewritten as   
		\begin{align}\label{ZF4} 
		&\max\limits_{\widetilde{\bm{F}}_{A_k},\widetilde{\bm{F}}_{D_k}} \Re\big\{
		\text{Tr}\big(\widetilde{\bm{F}}_{A_k}\widetilde{\bm{F}}_{D_k} (\bm{F}_k^{\rm{ZF}})^{\rm H}
		\big)\big\} , \nonumber\\
		& ~~~ \text{s.t.} ~~~ \vert[\bm{F}_{A_k}]_{n,m}\vert=1, 
		\Vert\widetilde {\bm{F}}_{D_k}\Vert_F^2=P_k, ~\forall n,m,k.
		\end{align}  
		
		Although the problem \eqref{ZF4} is much simplified compared to the problem \eqref{ZF3},
		it is still challenging to directly design the analog precoder $\bm{F}_{A_k}$ in the
		unit-modulus space. We resort to an iterative-PP based method with two key ingredients:
		unconstrained optimal analog precoder and alternating minimization. The unconstrained
		optimal analog precoder $\bm{F}_{A_k}^{\rm{Unc}}$ to the problem \eqref{ZF4} using
		majorization theory \cite{sadek2007leakage} is summarized in the following proposition. 
		
		\begin{proposition}\label{p12}  
			The unconstrained optimal analog precoder $\bm{F}_{A_k}^{\rm{Unc}}$ to the problem
			\eqref{ZF3} is
			\begin{align}\label{p13} 
			\bm{F}_{A_k}^{\rm{Unc}} =& \bm{U}_k^{\rm{ZF}}\bm{\Lambda}_{\bm{F}_{A_k}}
			\bm{V}_{\bm{F}_{A_k}} ,
			\end{align} 
			where the unitary matrix $\bm{U}_k^{\rm{ZF}}$ comes from the SVD $\bm{F}_k^{\rm{ZF}}
			\! =\! \bm{U}_{k}^{\rm{ZF}}\bm{\Lambda}_k^{\rm{ZF}}\bm{V}_k^{\rm{ZF}}$. Moreover,
			both the diagonal matrix $\bm{\Lambda}_{\bm{F}_{A_k}}$ and the unitary matrix
			$\bm{V}_{{\bm{F}}_{A_k}}$ can be arbitrarily  chosen. 
		\end{proposition}
		
		\begin{proof}
			Define the SVDs $\bm{F}_{A_k}\!\! =\! \bm{U}_{\bm{F}_{A_k}}\!\bm{\Lambda}_{\bm{F}_{A_k}}\!
			\bm{V}_{\bm{F}_{A_k}}$, $\bm{F}_{D_k}\!\! = \!\bm{U}_{\bm{F}_{D_k}}\!\bm{\Lambda}_{\bm{F}_{D_k}}\!
			\bm{V}_{\bm{F}_{D_k}}$ and $\widetilde{\bm{F}}_{D_k}\!\! =\! \bm{U}_{\widetilde{\bm{F}}_{D_k}}\!
			\bm{\Lambda}_{\widetilde{\bm{F}}_{D_k}}\! \bm{V}_{\widetilde{\bm{F}}_{D_k}}$, where
			$\{\bm{U}_{\bm{F}_{A_k}},\bm{U}_{\bm{F}_{D_k}},\bm{U}_{\widetilde{\bm{F}}_{D_k}}\}$ and
			$\{\bm{V}_{\bm{F}_{A_k}},\bm{V}_{\bm{F}_{D_k}},\bm{V}_{\widetilde{\bm{F}}_{D_k}}\}$ are
			the sets of unitary matrices, while $\{\bm{\Lambda}_{\bm{F}_{A_k}},\bm{\Lambda}_{{\bm{F}}_{D_k}},
			\bm{\Lambda}_{\widetilde{\bm{F}}_{D_k}}\}$ are the corresponding diagonal matrices with diagonal
			elements arranged in a decreasing order. {It is observed  from \eqref{ZF4} that $\widetilde{\bm{F}}_{D_k}$
				subject to the power constraint $\Vert\widetilde {\bm{F}}_{D_k}\Vert_F^2\! =\! P_k$ is 
				unitarily invariant.} In other words, both the unitary matrices $\bm{U}_{\widetilde{\bm{F}}_{D_k}}$
			and $\bm{V}_{\widetilde{\bm{F}}_{D_k}}$ are unconstrained. In addition, observing from
			$\widetilde{\bm{F}}_{A_k}\! =\! \bm{F}_{A_k}(\bm{F}_{A_k}^{\rm H}\bm{F}_{A_k})^{-\frac{1}{2}}
			\! =\! \bm{U}_{\bm{F}_{A_k}}\bm{V}_{\bm{F}_{A_k}}$,  {we find that the diagonal matrix
				$\bm{\Lambda}_{\bm{F}_{A_k}}$ actually has no effect on the maximum value of the
				objective function in \eqref{ZF4}, which is also applicable  to the problem  \eqref{ZF3}. Further by applying  \cite[B.2.~Theorem (Fan,1951)]{booktypical} to the problem  \eqref{ZF4},}
			the unconstrained optimal analog precoder $\bm{F}_{A_k}^{\rm{Unc}}$ to the problem
			\eqref{ZF3} is readily derived as  \eqref{p13}, where $\bm{\Lambda}_{\bm{F}_{A_k}}$
			and $\bm{V}_{{\bm{F}}_{A_k}}$ are arbitrarily chosen.
		\end{proof}
		\vspace{-2mm}
		Based on Proposition~\ref{p12},  we then aim to find an unit-modulus analog precoder
		$\bm{F}_{A_k}$ with the minimum Euclidean distance to the unconstrained optimal
		$\bm{F}_{A_k}^{\rm{Unc}}$, which is formulated as 
		\begin{align}\label{ZFp1} 
		&\min\limits_{\bm{\Lambda}_{\bm{F}_{A_k}},\bm{V}_{\bm{F}_{A_k}},\bm{F}_{A_k}} \Vert
		\bm{F}_{A_k}^{\rm{Unc}}\!-\!{\bm{F}}_{A_k} \Vert_F^2\! =\!\Vert \bm{U}_{k}^{\rm{ZF}}
		\bm{\Lambda}_{\bm{F}_{A_k}}\bm{V}_{\bm{F}_{A_k}}\! - \!\bm{F}_{A_k} \Vert_F^2 , \nonumber\\
		&~~~~~~~~
		\text{s.t.} ~~~~~ \vert[\bm{F}_{A_k}]_{n,m}\vert=1,~~~\forall n,m, k .
		\end{align}
		
		Since the diagonal matrix $\bm{\Lambda}_{\bm{F}_{A_k}}$ has no effect on the
		problem \eqref{ZF3}, we consider the unconstrained diagonal matrix
		$\bm{\Lambda}_{\bm{F}_{A_k}}$. Although the problem \eqref{ZFp1} is not jointly
		convex w.r.t $\{\bm{\Lambda}_{\bm{F}_{A_k}},\bm{V}_{\bm{F}_{A_k}},\bm{F}_{A_k}\}$,
		it is a `semi-convex' problem, in which the closed-form solution of each variable
		is easily obtained when fixing all the others, thus enabling alternating
		optimization. Specifically, given $\bm{\Lambda}_{\bm{F}_{A_k}}$ and
		$\bm{V}_{\bm{F}_{A_k}}$, the optimal analog precoder $\bm{F}_{A_k}$ to the problem
		\eqref{ZFp1} can be obtained via PP
		\begin{align}\label{ZFp2} 
		\bm{F}_{A_k} =& e^{\textsf{j}\arg\big(\bm{U}_{k}^{\rm{ZF}}\bm{\Lambda}_{\bm{F}_{A_k}}
			\bm{V}_{\bm{F}_{A_k}}\big)}, ~ \forall k.
		\end{align}
		By fixing $\bm{F}_{A_k}$ and $\bm{V}_{\bm{F}_{A_k}}$, after some algebraic manipulations,
		the optimal diagonal matrix $\bm{\Lambda}_{\bm{F}_{A_k}}$ to the problem \eqref{ZFp1}
		is designed by solving the optimization
		\begin{align}\label{ZFp3} 
		&\max\limits_{\bm{\Lambda}_{\bm{F}_{A_k}}}~~ \Re\big\{ \text{Tr}\big(\bm{V}_{\bm{F}_{A_k}}
		\bm{F}_{A_k}^{\rm H}\bm{U}_{k}^{\rm{ZF}}\bm{\Lambda}_{\bm{F}_{A_k}}\big)\big\} , \nonumber\\&~~
		\text{s.t.} ~~~~ \vert[\bm{F}_{A_k}]_{n,m}\vert=1,\forall n,m,k,
		\end{align}
		which has the closed-form solution
		\begin{align}\label{eq40} 
		\big[\bm{\Lambda}_{\bm{F}_{A_k}}\big]_{i,i}\! =\! \Re\big\{ \big[\bm{V}_{\bm{F}_{A_k}}
		\bm{F}_{A_k}^{\rm H}\bm{U}_{k}^{\rm{ZF}}\big]_{i,i}\big\}, ~ i\!=\!1,\cdots ,N_{t_k}^{RF} .
		\end{align}
		Finally, for the fixed $\bm{\Lambda}_{\bm{F}_{A_k}}$ and $\bm{F}_{A_k}$, the
		optimal unitary matrix $\bm{V}_{\bm{F}_{A_k}}$ is given by
		\begin{align}\label{ZFp4} 
		\bm{V}_{\bm{F}_{A_k}} =& \bm{V}_{A_k}^{\rm H}\bm{U}_{A_k}^{\rm H},~ \forall k,
		\end{align}
		where the unitary matrices $\bm{V}_{A_k}$ and $\bm{U}_{A_k}$ come from the SVD
		$\bm{F}_{A_k}^{\rm H}\bm{U}_{k}^{\rm{ZF}}\bm{\Lambda}_{\bm{F}_{A_k}}\! =\! \bm{U}_{A_k}
		\bm{\Lambda}_{A_k}\bm{V}_{A_k}$.
		
		Through alternating optimization among  \eqref{ZFp2}, \eqref{eq40} and \eqref{ZFp4},
		the iterative PP-based unit-modulus analog precoder $\bm{F}_{A_k}$  can be finally
		obtained. Then by applying Lagrangian multiplier method to the problem \eqref{ZF3},
		the optimal digital precoder $\bm{F}_{D_k}$  given $\bm{F}_{A_k}$ is expressed as   
		\begin{align}\label{ZF7} 
		\bm{F}_{D_k} =& \frac{\sqrt{P_k} \bm{F}_{A_k}^{\rm H}\bm{F}_{k}^{\rm{ZF}}}
		{\Vert(\bm{F}_{A_k}^H\bm{F}_{A_k})^{\frac{1}{2}}\bm{F}_{A_k}^{\rm H}\bm{F}_{k}^{\rm{ZF}}
			\Vert_F}, ~ \forall k.
		\end{align} 
		
		Even when the distance between the hybrid precoder and the fully-digital BD-ZF
		precoder is minimized, we still cannot guarantee the hybrid precoder's capability
		of realizing zero inter-user interference, since it may not be exactly located in
		the null-space of the corresponding leakage channels. However, the effectiveness
		of the above iterative-PP hybrid precoder design on suppressing the inter-user
		interference is demonstrated in the following proposition. 
		
		\begin{proposition}\label{Le2} 
			For mmWave channels, once the hybrid precoder $\bm{F}_{A_k}{\bm{F}}_{D_k}$
			is obtained from \eqref{ZFp2} and \eqref{ZF7}, the resultant inter-user interference
			to the $i$th receiver, where $i\neq k$, satisfies
			\begin{align}\label{eq43}
			\lim_{N_{t_k} \to +\infty} \bm{H}_{i,k}{\bm{F}}_{A_k}{\bm{F}}_{D_k} = \bm{0}, ~ \forall k=1,\cdots , K,~
			i\neq k .
			\end{align}
		\end{proposition}
		\vspace{-4mm}
		\begin{proof}
			See Appendix~\ref{APb}.
		\end{proof}
		Proposition~\ref{Le2} reveals that in large-scale mmWave scenarios, the above
		iterative-PP hybrid precoder also achieves the near-zero inter-user interference,
		like the fully-digital BD-ZF precoder. However, it may not work well in 
		Rayleigh channels due to the following two reasons. The fully-digital BD-ZF
		precoder with a strict restriction on the numbers of transmit and receive
		antennas may be infeasible in rich scattering scenarios, especially for a large
		number of transmit-receive pairs. Also Proposition~\ref{Le2} is not applicable
		to the channels without sparsity. These facts motivate us to propose another
		more general low-complexity hybrid precoder design in Section~\ref{lowB}.
		
		\subsubsection{MM-based hybrid combiner design}
		Given the hybrid precoder \eqref{ZFp2} and \eqref{ZF7}, the MMSE hybrid combiner
		design is then formulated as\vspace{-2mm}
		\begin{align}\label{cz1} 
		&\min\limits_{\bm{G}_{A_k},\bm{G}_{D_k}} \mathbb{E}\big[\Vert \bm{s}_k\!-\!\bm{G}_{D_k}^{\rm H}
		\bm{G}_{A_k}^{\rm H}\bm{y}_k\Vert^2\big]\!\! \overset{(a)}{=}\! \!\Vert \bm{R}_{\bm{y}_k}^{\frac{1}{2}}
		(\widehat{\bm{G}}_k\!-\!\bm{G}_{A_k}\bm{G}_{D_k}) \Vert_F^2 , \nonumber\\
		& ~~~~\text{s.t.} ~~~~
		\vert[\bm{G}_{A_k}]_{n,m}\vert=1 ,
		\end{align}
		where $\bm{R}_{\bm{y}_k}\! =\! \mathbb{E}[\bm{y}_k\bm{y}_k^{\rm H}]\! =\!\!
		\sum\nolimits_{i=1}^K \bm{H}_{k,i}\bm{F}_{A_i}\bm{F}_{D_i}\bm{F}_{D_i}^{\rm H}
		\bm{F}_{A_i}^{\rm H}\bm{H}_{k,i}^{\rm H}\! +\! \sigma_{n_k}^2\bm{I}_{N_{r_k}}$,
		and $\widehat{\bm{G}}_k\! =\! \bm{R}_{\bm{y}_k}^{-1}\bm{H}_{k,k}\bm{F}_{A_k}\bm{F}_{D_k}$,
		while the equality (a) holds by following the  similar derivations in
		\cite{el2014spatially}. Obviously, the problem \eqref{cz1} is more complicated 
		than the problem \eqref{ZF3} due to introducing $\bm{R}_{\bm{y}_k}$. Fortunately,
		the proposed MM-based alternating optimization is  still applicable, as shown
		below.
		
		When the analog  combiner $\bm{G}_{A_k}$ is fixed, the optimal digital combiner
		$\bm{G}_{D_k}$ has the closed-form
		\begin{align}\label{cz2} 
		\bm{G}_{D_k} =& \big(\bm{G}_{A_k}^{\rm H}\bm{R}_{\bm{y}_k}\bm{G}_{A_k}\big)^{-1}
		\bm{G}_{A_k}^{\rm H}\bm{R}_{\bm{y}_k}\widehat{\bm{G}}_{k} .
		\end{align}
		Given the fixed digital combiner $\bm{G}_{D_k}$, the MM  method is used to tackle
		the nonconvex problem \eqref{cz1} in terms of $\bm{G}_{A_k}$ by finding an
		appropriate majorized problem, which is
		\begin{align}\label{cz3} 
		\min\limits_{\bm{g}_{A_k}} \Re\{ \bm{g}_{A_k}^{\rm H}\widetilde{\bm{r}}_k\}, ~
		\text{s.t.} ~ \vert[\bm{g}_{A_k}]_{n}\vert\!=\!\!1 , \forall n\!=\!\!1,\cdots, N_{t_k}N_{t_k}^{RF} ,
		\end{align} 
		where $\widetilde{\bm{r}}_k\! =\! \big(\widetilde{\bm{R}}_k\! -\!
		\lambda_{\max}(\widetilde{\bm{R}}_k)\bm{I}_{N_{r_k}N_{r_k}^{RF}}\big)\bm{g}_{A_k}^{(l)}
		\! -\! \bm{r}_k$, $\bm{r}_k^{\rm H}\! =\! \text{vec}\big((\bm{G}_{D_k}
		\widehat{\bm{G}}_k^{\rm H}\bm{R}_{\bm{y}_k})^{\rm T}\big)^{\rm T}$ and
		$\widetilde{\bm{R}}_k\! =\! \big(\bm{G}_{D_k}\bm{G}_{D_k}^{\rm H}\big)^{\rm T}\!
		\otimes\! \bm{R}_{\bm{y}_k}$. The semi closed-form solution to the problem
		\eqref{cz3} is given by
		\begin{align}\label{cz4} 
		\bm{g}_{A_k} =& \text{vec}(\bm{G}_{A_k}) =  -e^{\textsf{j}\arg (\widetilde{\bm{r}}_k)} .
		\end{align}
		Due to the iterative nature between \eqref{cz2} and \eqref{cz4}, the hybrid combiner
		obtained better matches with the iterative-PP  hybrid precoder in \eqref{ZFp2} and
		\eqref{ZF7}.
		
		By integrating the sparse recovery problems \eqref{ZF3} and \eqref{cz1}, the proposed
		BD-ZF hybrid transceiver design is summarized in Algorithm~\ref{AA2}.
		
		{\renewcommand\baselinestretch{0.9}\selectfont
			\begin{algorithm}[tp!]
				\caption{Low-complexity BD-ZF/SLNR-Max hybrid transceiver designs}
				\label{AA2} 
				\begin{small}
					\begin{algorithmic}[1]
						\REQUIRE{BD-ZF/SLNR-Max fully-digital precoder $\bm{F}_k^{\rm{ZF}}\big/\bm{F}_k^{\rm{SL}}$,
							$\forall k$, derived from \eqref{ZF1}/\eqref{SL5}; initial analog precoder $\bm{F}_{A_k}^{(0)}$
							and combiner $\bm{G}_{A_k}^{(0)}$, $\forall k$, derived from \eqref{TWA2};
							outer iteration indexed $I_{t}=0$ and $I_{r}=0$.}
						\REPEAT
						\STATE{Fix $\bm{F}_{A_k}^{(I_t)}$, $\forall k$, calculate $\bm{F}_{D_k}^{(I_{t}+1)}$,
							$\forall k$, according to \eqref{ZF7}.}
						\STATE{Fix $\bm{F}_{D_k}^{(I_{t}+1)}$, $\forall k$, calculate $\bm{F}_{A_k}^{(I_{t})}$, 
							$\forall k$, using MM method as in Algorithm~\ref{AL1}.}
						\STATE{Set $I_t=I_t+1$.}
						\UNTIL{Objective function value of problem \eqref{ZF3} converges.}
						\STATE{Calculate normalized $\widehat{\bm{F}}_{D_k}^{(I_t)}=\frac{\sqrt{P_k}}
							{\Vert \bm{F}_{A_k}\bm{F}_{D_k}^{\rm{ZF}}\Vert_F} \bm{F}_{D_k}^{\rm{ZF}}$, $\forall k$.}
						\STATE{Fix $\bm{F}_{A_k}^{(I_{t})}$ and $\widehat{\bm{F}}_{D_k}^{(I_{t})}$, $\forall k$,
							calculate fully-digital combiner $\widehat{\bm{G}}_k$, $\forall k$.}
						\REPEAT
						\STATE{Fix $\bm{G}_{A_k}^{(I_{r})}$, $\forall k$, calculate $\bm{G}_{D_k}^{(I_{r}+1)}$,
							$\forall k$, according to \eqref{cz2}.}
						\STATE{Fix $\bm{G}_{D_k}^{(I_{r}+1)}$, $\forall k$, calculate $\bm{G}_{A_k}^{(I_{r})}$, 
							$\forall k$, using MM method as in Algorithm~\ref{AL1}.}
						\STATE{Set $I_r=I_r+1$.}
						\UNTIL{Objective function value of problem \eqref{cz1} converges.}
						\ENSURE{$\big\{\bm{G}_{D_k}^{(I_{r})},\bm{G}_{A_k}^{(I_{r})},\widehat{\bm{F}}_{D_k}^{(I_{t})},
							\bm{F}_{A_k}^{(I_{t})}$, $\forall k\big\}$.}
					\end{algorithmic}
				\end{small}
			\end{algorithm}
		}
		%
		\vspace{-3mm}
		\subsection{SLNR-Max hybrid transceiver design}\label{lowB}
		The drawbacks of the BD-ZF technique are the restriction on the number of antennas
		and the noise enhancement. We consider alternative design based on SLNR maximization.
		The SLNR of the $k$th transmitter is defined as the ratio of the received signal
		power at the desired $k$th receiver to the interference (leakage) at other receivers
		plus noise power \cite{sadek2007leakage}
		\begin{align}\label{SL2} 
		\text{SLNR}_k  \! \!=   \!\!\frac{\text{Tr}\big((\bm{F}_k^{\rm{SL}})^{\rm H}\bm{H}_{k,k}^{\rm H}
			\bm{H}_{k,k}\bm{F}_{k}^{\rm{SL}}\big)}{\text{Tr}\big(\sigma_{n_k}^2\bm{I}_{N_{t_k}^{RF}}
			\!+  \!\sum\nolimits_{i\neq k} (\bm{F}_{k}^{\rm{SL}})^{\rm H}\bm{H}_{k,i}^{\rm H}\bm{H}_{k,i}
			\bm{F}_k^{\rm{SL}}\big)} .
		\end{align}
		Then the SLNR-Max fully-digital precoder for each transmit-receiver pair is designed
		as \cite{sadek2007leakage}\vspace{-1mm}
		\begin{align}\label{SL3} 
		\bm{F}_k^{\rm{SL}}  \! \!= \! \! \arg \max~\! \text{SLNR}_k , ~\! {\rm{s.t.}} ~
		\text{Tr}\big((\bm{F}_k^{\rm{SL}})^{\rm H}\bm{F}_k^{\rm{SL}}\big) \!\le  \!P_k , \forall k.
		\end{align}
		Define the generalized EVD for the  matrix pencil $\big(\bm{H}_{k,k}^{\rm H}
		\bm{H}_{k,k},\frac{N_{r_k}\sigma_{n_k}^2}{P_k}\bm{I}_{N_{t_k}^{RF}}\! +\!
		\sum\nolimits_{i\neq k}\bm{H}_{k,i}^{\rm H} \bm{H}_{k,i}\big)$ as\vspace{-1mm}
		\begin{align}\label{SL4} 
		\left\{  \! \! \!\begin{array}{l}
		\bm{T}_k^{\rm H}\bm{H}_{k,k}^{\rm H}\bm{H}_{k,k}\bm{T}_k \!=  \! \!\bm{\Sigma}_{k}  \!= \! \!
		\text{diag}[\sigma_{k,1},\cdots ,\sigma_{k,N_{t_k}^{RF}}] , \\
		\bm{T}_k^{\rm H}\!\Big( \frac{N_{r_k}\sigma_{n_k}^2}{P_k}\bm{I}_{N_{t_k}^{RF}}  \! \!+  \!\!
		\sum\nolimits_{i\neq k} \bm{H}_{k,i}^{\rm H}\bm{H}_{k,i}\Big)\!\bm{T}_k \! \! = \! \! \bm{I}_{N_{t_k}^{RF}}, \forall k,
		\end{array}\right.
		\end{align} 
		where the columns of $\bm{T}_k\! \in\! \mathbb{C}^{N_{t_k}^{RF}\times N_{t_k}^{RF}}$
		and the diagonal elements of $\bm{\Sigma}_{k}$ are the generalized eigenvectors and
		eigenvalues, respectively. Then the optimal SLNR-Max fully-digital precoder is given by
		\begin{align}\label{SL5} 
		\bm{F}_k^{\rm SL} \!\!=\!\! \sqrt{\frac{P_k}{\text{Tr}\big(\bm{T}_k^{\rm H}(:,1\!:\!N_{s_k})
				\bm{T}_k(:,1\!:\!N_{s_k})\big)}} \bm{T}_k(:,1\!:\!N_{s_k}) .
		\end{align}
		Similarly to the BD-ZF hybrid design, we formulate the SLNR-Max hybrid design by
		minimizing the Euclidean distance between \eqref{SL5} and the hybrid counterpart as\vspace{-1mm}
		\begin{align}\label{SL6} 
		&\min\limits_{\bm{F}_{A_k},\widetilde{\bm{F}}_{D_k}} \Vert \bm{F}_k^{\rm{SL}} -
		\bm{F}_{A_k}\bm{F}_{D_k}\Vert_F^2 , \nonumber\\
		&~~~ \text{s.t.} ~~~ \vert[\bm{F}_{A_k}]_{n,m}\vert=1,
		\Vert\bm{F}_{A_k}\bm{F}_{D_k}\Vert_F^2 =P_k,  \forall n,m, k .
		\end{align}
		
		The problem \eqref{SL6} can be effectively solved following the same approach
		of solving the problem \eqref{ZF3} by replacing $\bm{F}_k^{\rm{ZF}}$ with 
		$\bm{F}_k^{\rm{SL}}$, i.e., the MM-based alternating optimization is applicable.
		Additionally, once the SLNR-Max hybrid precoder for each transmit-receive pair
		is obtained, the corresponding MMSE hybrid combiner design can be independently
		carried out as in \eqref{cz1}. This  SLNR-Max hybrid design is also summarized
		in Algorithm~\ref{AA2}.
		\vspace{-4mm}
		\subsection{Partially-connected hybrid transceiver structure}\label{S5.3}
		In the partially-connected structure, each RF chain at both ends is only connected
		with a part of the antenna array. Specifically, at the $k$th transmitter (receiver),
		each RF chain is only connected with $N_{t_k}/N_{t_k}^{RF}$ ($N_{r_k}/N_{r_k}^{RF}$)
		antennas, and thus the analog precoder $\bm{F}_{A_k}$ and combiner $\bm{G}_{A_k}$,
		$\forall k$, can be expressed by the following block matrices\vspace{-1mm}
		{\begin{align}\label{cc} 
			&\bm{F}_{A_k} =\text{BLkdiag}[\bm{p}_{k_1}, \bm{p}_{k_2} \cdots\bm{p}_{k_{N_{t_k}^{RF}}}  ],\nonumber\\
			&\bm{G}_{A_k} =\text{BLkdiag}[\bm{q}_{k_1}, \bm{q}_{k_2} \cdots\bm{q}_{k_{N_{r_k}^{RF}}} ],
			\end{align}}where the unit-modulus entries $\vert [\bm{p}_{i_k}]_{m_k}\vert\! =\!1$, $\forall i_k
		\! =\! 1,\cdots,N_{t_k}^{RF}$, $\forall m_k\! =\! 1,\cdots,N_{t_k}/N_{t_k}^{RF}$, and
		$\vert [\bm{q}_{j_k}]_{n_k}\vert\! =\! 1$, $\forall j_k\! =\! 1,\cdots, N_{r_k}^{RF}$,
		$\forall n_k\! =\! 1,\cdots,N_{r_k}/N_{r_k}^{RF}$, are imposed. Benefited from the 
		block diagonal structures of the analog precoder and combiner, the MM-based alternating
		optimization can be directly applied to the {WMMSE} problem \eqref{eq6} to obtain the
		locally optimal solution without requiring the approximation on analog precoder as in
		Section~\ref{MMALT}. More importantly, due to the sparsity of the partially-connected
		structure, the MM-based analog  {precoder  and combiner designs exhibit much lower
			complexity than that of Section~\ref{S3}.}
		
		\subsubsection{Semi closed-form digital precoder $\bm{F}_{D_k}^{\rm{Par}}$}
		
		Based on the partially-connected  structure \eqref{cc}, we can re-express
		$\widetilde{\bm{F}}_{A_k}$ and $\widetilde{\bm{F}}_{D_k}$ as $\widetilde{\bm{F}}_{A_k}
		\! =\! \sqrt{{N_{t_k}^{RF}}/{N_{t_k}}}\bm{F}_{A_k}$ and $\widetilde{\bm{F}}_{D_k}
		\! =\! \sqrt{{N_{t_k}}/{N_{t_k}^{RF}}}\bm{F}_{D_k}$, respectively, which are then
		substituted into \eqref{eq16} to obtain the semi closed-form digital precoder
		$\bm{F}_{D_k}^{\rm{Par}}$\vspace{-1mm}
		{
			\begin{align}\label{Par1} 
			\!\!\!\!\begin{array}{l}
			\bm{F}_{D_k}^{\rm{Par}} \!=\! \frac{1} {\sqrt{N_{t_k}}} \big(\sum\limits_{i=1}^K 
			\bm{L}_{i,k}^{\rm H} \bm{W}_i\bm{L}_{i,k} \!+\! \beta_k^{'}\bm{I}_{N_{t_k}^{RF}}\big)^{-1}
			\!\bm{L}_{k,k}^{\rm H}\bm{W}_k .
			\end{array}
			\end{align} where the determination of scalar $\beta_k^{'}$ is similar to  $\beta_k$.}
		\subsubsection{Semi closed-form analog precoder $\bm{F}_{A_k}^{\rm{Par}}$} 
		Given $\{\bm{F}_{D_k}^{\rm{Par}},$ $\forall k\}$, we firstly define the following
		auxiliary parameters for optimizing the partially-connected analog precoder  $\bm{F}_{A_k}^{\rm{Par}}$, which are\vspace{-1mm}
		\begin{align}\label{eq55}
		\left\{\!\!\!\! \begin{array}{l}
		\widehat{\bm{A}}_k\! =\!\!\frac{N_{t_k}^{RF}}{N_{t_k}}\! \left[\! \begin{array}{ccc}
		\widehat{\bm{A}}_k^{1,1} & \cdots & \widehat{\bm{A}}_k^{1,N_{t_k}^{RF}} \\
		\vdots & \ddots & \vdots \\ \widehat{\bm{A}}_k^{N_{t_k}^{RF},1} & \cdots & 
		\widehat{\bm{A}}_k^{N_{t_k}^{RF},N_{t_k}^{RF}} \end{array}\! \right]\!\! \in\!
		\mathbb{C}^{N_{t_k}\!\times\! N_{t_k}} , \\
		\widehat{\bm{A}}_k^{l,q}\! \!=\!\! \widetilde{\bm{A}}_k\Big(\widetilde{l}:\widetilde{l}\!
		+\! \frac{N_{t_k}}{N_{t_k}^{RF}}\! - \!1,\widetilde{q}:\tilde{q}\!\! +\!\!
		\frac{N_{t_k}}{N_{t_k}^{RF}}\!\! -\! \!1\Big) ,\\ 
		\bar{\bm{f}}_{A_k}\!\! =\!\!\big[\bm{p}_1^{\rm T}  \cdots
		\bm{p}_{N_{t_k}^{RF}}^{\rm T}\big]^{\rm T}\!\!\! \in\!\!\mathbb{C}^{N_{t_k} } , 
		\widehat{\bm{a}}_k\!\! =\!\! \big[\widehat{\bm{a}}_k^1 \cdots
		\widehat{\bm{a}}_k^{N_{t_k}^{RF}}\big]\!\! \in\! \mathbb{C}^{N_{t_k} } , \\
		\widehat{\bm{a}}_k^l\!\! =\!\! \sqrt{\frac{N_{t_k}^{RF}}{N_{t_k}}}
		\bm{a}_k\big(\widetilde{l}:\widetilde{l}\! + \!\frac{N_{t_k}}{N_{t_k}^{RF}}\! -\! 1\big) , 
		\widetilde{l}\! =\! (l\! -\! 1)\big(\frac{N_{t_k}}{N_{t_k}^{RF}}\!\! +\! N_{t_k}\big)\!\! +\!\!1,\\
		\widetilde{q}\! =\! (q\! -\! 1)\big(\frac{N_{t_k}}{N_{t_k}^{RF}}\! \!+\!\! N_{t_k}\big)\! +\! 1,~
		\forall l,q=1,\cdots,N_{t_k}^{RF} .
		\end{array} \right.
		\end{align}
		Following the similar derivations of \eqref{eq19}, the partially-connected analog
		precoder $\bm{F}_{A_k}^{\rm{Par}}$ for each transmit-receive pair is independently
		designed as
		\begin{align}\label{eqP2} 
		&\min\limits_{\bar{\bm{f}}_{A_k}}~ \bar{\bm{f}}_{A_k}^{\rm H}\widehat{\bm{A}}_k
		\bar{\bm{f}}_{A_k} - 2\Re\big\{\widehat{\bm{a}}_k^{\rm H}\bar{\bm{f}}_{A_k}\big\} , \nonumber\\
		&~
		{\rm{s.t.}} ~~ \vert[\bar{\bm{f}}_{A_k}]_n\vert=1, n=1,\cdots ,N_{t_k}.
		\end{align}
		
		Recalling Lemma~\ref{Le1}, the majorized counterpart of the problem \eqref{eqP2} at
		$\bar{\bm{f}}_{A_k}^{(l)}$ is formulated as
		\begin{align}\label{eq57}
		\min\limits_{\bar{\bm{f}}_{A_k}} ~ {\Re\big\{ \vec{\bm{a}}_k^{\rm H}
			\bar{\bm{f}}_{A_k}\big\} }, ~ \text{s.t} ~ \vert [ \bar{\bm{f}}_{A_k}]_n\vert\!=\!1,
		\forall n\!=\!1,\cdots, N_{t_k} ,
		\end{align}  
		where $\vec{\bm{a}}_k\! =\! \big(\widehat{\bm{A}}_k\! -\!
		\lambda_{\rm{\max}}(\widehat{\bm{A}}_k)\bm{I}_{N_{t_k}}\big)\widetilde{\bm{f}}_{A_k}^{(l)}
		\! -\! \widehat{\bm{a}}_k$, and the semi closed-form solution is obtained as \vspace{-1mm}
		\begin{align}\label{eqP5} 
		\bar{\bm{f}}_{A_k} =& \text{vec}(\bm{F}_{A_k}^{\rm{Par}}) = -e^{\textsf{j}\arg (\vec{\bm{a}}_k)} .
		\end{align}
		
		\subsubsection{Semi closed-form analog combiner $\bm{G}_{A_k}^{\rm{Par}}$}
		Similarly to solving \eqref{eq23}, by defining\vspace{-1mm}
		\begin{align}
		\left\{\!\!\begin{array}{l}
		\widehat{\bm{N}}_k\!\! =\!\!\! \left[\begin{array}{ccc} \widehat{\bm{N}}_k^{1,1} & \cdots & 
		\widehat{\bm{N}}_k^{1,N_{r_k}^{RF}} \\ \vdots & \ddots & \vdots \\ \widehat{\bm{N}}_k^{N_{r_k}^{RF},1}
		& \cdots & \widehat{\bm{N}}_k^{N_{r_k}^{RF},N_{r_k}^{RF}} \end{array}\right]\!\!\in\!
		\mathbb{C}^{N_{r_k}\times N_{r_k}} \\
		\widehat{\bm{N}}_k^{l,q}\!\! =\!\! \widetilde{\bm{N}}_k\big(\widetilde{l}:\widetilde{l}\! +\!
		\frac{N_{r_k}}{N_{r_k}^{RF}}\!\! -\!\! 1, \widetilde{q}:\widetilde{q}\! +\!
		\frac{N_{r_k}}{N_{r_k}^{RF}}\! - \!1\big) ,
		\end{array} \right.
		\notag 
		\end{align}
		\begin{align}\label{eq59}
		\left\{\!\!\!\!\begin{array}{l}
		\bar{\bm{g}}_{A_k}\!\! \!=\!\!\big[\bm{q}_1^{\rm T} ~\cdots
		\bm{q}_{N_{r_k}^{RF}}^{\rm T}\big]^{\rm T}\!\!\! \in\! \!\mathbb{C}^{N_{r_k}}, 
		\widehat{\bm{d}}_k\!\! =\!\!\big[\widehat{\bm{d}}_k^1\cdots \widehat{\bm{d}}_k^{N_{r_k}^{RF}}\big]\!\!
		\in\!\! \mathbb{C}^{N_{r_k}} ,\\
		\widehat{\bm{d}}_k^ {l}\!\!= \!\bm{d}_k \big(\widetilde{l}:\widetilde{l}+\frac{N_{r_k}}{N_{r_k}^{RF}}-1\big), 
		\widetilde{l}\! =\! (l\! \!-\!\! 1)\big(\frac{N_{r_k}}{N_{r_k}^{RF}}\! +\! N_{r_k}\big)\!\! +\!\! 1,\\
		\widetilde{q}\! =\! (q\! -\! 1)\big(\frac{N_{r_k}}{N_{r_k}^{RF}}\! +\! N_{r_k}\big)\! +\! 1,
		~\forall l,q\! =\! 1,\cdots,N_{r_k}^{RF} ,
		\end{array} \right.
		\end{align}
		the independent design of partially-connected analog combiner
		$\bm{G}_{A_k}^{\rm{Par}}$ for each transmit-receive pair can be formulated as
		\begin{align}\label{eqP7} 
		&\min\limits_{\bar{\bm{g}}_{A_k}}~ \bar{\bm{g}}_{A_k}^{\rm H}\widehat{\bm{N}}_k\bar{\bm{g}}_{A_k}
		- 2\Re\big\{\widehat{\bm{d}}_k{\rm ^H}\bar{\bm{g}}_{A_k}\big\} , \nonumber\\
		&~ {\rm{s.t.}} ~
		\vert [\bar{\bm{g}}_{A_k}]_n\vert=1, \forall n=1,\cdots, N_{r_k}.
		\end{align} 
		Also, the majorized counterpart of the problem \eqref{eqP7} at $\bar{\bm{g}}_{A_k}^{(l)}$
		can be expressed as \vspace{-1mm}
		\begin{align}\label{eqP8} 
		\min\limits_{\bar{\bm{g}}_{A_k}} ~ {\Re\big\{\vec{\bm{d}}_k^{\rm H}
			\bar{\bm{g}}_{A_k}\big\}}, ~ \text{s.t} ~ \vert [\bar{\bm{g}}_{A_k}]_n\vert=1,
		\forall n=1,\cdots, N_{r_k}.
		\end{align} 
		where $\vec{\bm{d}}_k\! =\! \big(\widehat{\bm{N}}_{k}\! -\! \lambda_{\max}(\widehat{\bm{N}}_k)
		\bm{I}_{N_{r_k}}\big)\bar{\bm{g}}_{A_k}^{(l)}\! -\! \widehat{\bm{d}}_k$, and the semi
		closed-form solution is derived as 
		\begin{align}\label{eqP9} 
		\bar{\bm{g}}_{A_k} =& \text{vec}\big(\bm{G}_{A_k}^{\rm{Par}}\big)
		= -e^{\textsf{j}\arg\big(\vec{\bm{d}}_k\big)}, ~ \forall k.
		\end{align}
		
		\subsubsection{Semi closed-form digital combiner $\bm{G}_{D_k}^{\rm{Par}}$ and weighting matrix $\bm{W}_{k}^{\rm{Par}}$}
		
		The optimal digital combiner $\bm{G}_{D_k}$ for the WMMSE problem
		\eqref{eq6} under this partially-connected structure is also Wiener filter, which
		has the same form as \eqref{eq10}. {Moreover}, the optimal weighing matrix
		$\bm{W}_{k}^{\rm{Par}}$ can be similarly derived as \eqref{eqw}. 
		
		Observing from {\eqref{eq57}} and \eqref{eqP8} that this partially-connected structure
		simplifies the analog precoder and combiner design due to the reduced number of
		optimization variables, { and also makes the  proposed MM-based alternating optimization directly applicable without the assumption in large-scale MIMO regime.}
		In a nutshell, the proposed MM-based hybrid design is well suited for this
		partially-connected structure.
		\vspace{-3mm}
		\section{Convergence of the Proposed Algorithms and Complexity Analysis}
		
		We firstly study the convergence of the proposed MM-based alternating optimization
		(MM-Alt-Opt). It is obvious that the objective function of the problem \eqref{eq8}
		is continuously differentiable and the constraint set is closed, bounded and
		separable in terms of optimization variables $\{\bm{G}_{A_k},\bm{F}_{A_k},
		\widetilde{\bm{F}}_{D_k}$, $\forall k\}$. In fact, the proposed MM-Alt-Opt for
		solving the problem \eqref{eq8} is a combination of the BCD and MM  methods, in
		which the uniquely optimal solutions of the blocks $\{\widetilde{\bm{F}}_{D_k},
		\forall k\}$, $\{\bm{G}_{D_k},\forall k\}$ and $\{\bm{W}_k,\forall k\}$ are
		available and the stationary solutions of the blocks $\{\bm{F}_{A_k},\forall k\}$
		and $\{\bm{G}_{A_k},\forall k\}$ are obtained using the MM method
		\cite{sun2017majorization}. Referring to \cite[Theorem~4.3]{jacobson2007expanded},
		since at least the stationary point for each block update is guaranteed, the
		proposed MM-Alt-Opt converges to a stationary point of the problem \eqref{eq8}.
		However, due to the adopted approximation on analog precoder, i.e.,
		$\widetilde{\bm{F}}_{A_k}\! \approx\! \frac{1}{\sqrt{N_{t_k}}}\bm{F}_{A_k}$, in
		\eqref{eq8}, this stationary point is actually a suboptimal solution to the
		original sum rate maximization problem \eqref{eq4}, but with an asymptotically
		optimal performance for large-scale MIMO regime according to Proposition~\ref{p1}.

		Next we analyze the computational complexity of the proposed MM-Alt-Opt, { PP-based  two-stage hybrid design (Hybrid PP-Two-Stage),} BD-ZF
		and SLNR-Max based hybrid designs {( Hybrid BD-ZF/SLNR-Max)}, in comparison with the classical OMP scheme
		\cite{el2014spatially}. To simplify the analysis, we consider that $N_t\! =\! N_{t_k}$,
		$N_r\! =\! N_{r_k}$ and $N_{RF}\! =\! N_{t_k}^{RF}\! =\! N_{r_k}^{RF}\! =\! N_{s_k}$,
		$\forall k$. In the OMP scheme, the length of codebooks for analog precoder (combiner)
		design is set to $L_{\rm c}$, and $N_t\! >\! N_r\! \gg \! L_{\rm c}\! >\! N_{RF}$ is
		assumed. We focus on the complexity of major computational steps, in which the low-order
		terms are omitted, and then the total complexity is added. 
		
		Let $I_W$ and $I_M$ be the numbers of outer and inner iterations, respectively, for
		the MM-based methods, including the MM-Alt-Opt and { Hybrid BD-ZF/SLNR-Max}. Observe
		from Algorithm~\ref{AL1} that in one outer iteration of the MM-Alt-Opt, the
		computational cost is mainly from the MM-based analog precoder design with the
		complexity on the order of $\textsf{O}(I_M N_t^2 N_{RF}^2)$ per transmit-receive
		pair. The total complexity of the MM-Alt-Opt is obviously linear w.r.t. the number of
		outer iterations $I_W$ and the number of communication pairs $K$. The similar analysis
		is  applicable to the partially-connected hybrid transceiver case {(Hybrid ParTxRx)}. While for the { Hybrid PP-Two-Stage},
		the complexity primarily comes from the selection of analog precoder and combiner
		based on the SVD of $N_r\! \times\! N_t$ channel matrix for each transmit-receive pair.
		The complexity of designing $\bm{F}_{D_k}$ and $\bm{G}_{D_k}$, which involves an
		iterative loop with $I_O$ iterations, is much smaller by comparison. This yields the
		total complexity of $\textsf{O}(K N_t^2 N_r)$. For the { Hybrid BD-ZF/SLNR-Max}, by defining $I_P$ as the number of iterations for the iterative-PP method,
		the hybrid precoder  design has the complexity $\textsf{O}(K I_P N_t^2 N_{RF})$,
		while the MM-based analog combiner design has the complexity
		$\textsf{O}(K I_W I_M N_r^2 N_{RF}^2)$. Hence the total complexity of this scheme is
		$\textsf{O}(K I_P N_t^2 N_{RF})$ {for the large $N_t$.} The OMP scheme involves an exhaustive search for
		both analog precoder and combiner from the predefined codebooks and large-scale
		matrix multiplication, yielding the  total complexity $\textsf{O}(K I_B N_t^3)$,
		where $I_B$ is the number of iterations for finding the WMMSE digital precoder and
		combiner.
		
		\vspace*{-2mm}
		\section{Simulation Results}~\label{SIM}
		\vspace*{-6mm}
		
		Unless otherwise stated, $K\! =\! 2$ transceiver pairs are used. Each transmitter
		deploys $N_t\! =\! 64$ antennas with $N_t^{RF}\! =\! 4$ RF chains to send $N_s\!
		=\! 4$ data streams to its receiver, which has $N_r\! =\! 16$ antennas and
		$N_t^{RF}\! =\! 4$ RF chains. The RF phase shifters with infinite resolution are
		assumed. Both the Rayleigh and mmWave channels are considered. For the normalized
		Rayleigh channel, the elements of all channel matrices are distributed according
		to $\mathcal{CN}(0,1)$. For the normalized mmWave channel, the propagation
		environment with $L_k\! =\! L\! =\! 10$ scatters, $\forall k$, is considered, in
		which the AOA and AOD of each path are uniformly distributed in $[0, ~ 2\pi ]$,
		while the pathloss factors $\alpha_k\! =\! \alpha$, $\forall k$, with $\alpha$
		obeying $\mathcal{CN}(0,1)$. By assuming the same transmit power $P_k\! =\! P$ and
		the same noise power $\sigma_{n_k}^2\! =\! \sigma_n^2$ at all transmitters and
		receivers, respectively, the received SNR becomes $\text{SNR}\! =\!
		\frac{P}{\sigma_n^2 }$. All the results are obtained by averaging over 100 channel
		realizations.
		
		In this work, we propose various hybrid transceiver designs, including the
		\textbf{MM-Alt-Opt}, the \textbf{Hybrid PP-Two-Stage}, the
		\textbf{Hybrid BD-ZF/SLNR-Max}, and the partially-connected hybrid structure of
		\textbf{Hybrid-ParTxRx}. In fact, there is another scheme which only considers
		the partially-connected hybrid structure at transmitter, and we call this scheme
		\textbf{Hybrid-ParTx}. The sum rate performance of these proposed designs are
		compared with that of the following baselines:
		
		\textbf{Hybrid OMP} \cite{el2014spatially}: The sparse reconstruction of the hybrid
		precoder and combiner of each transmit-receive pair is realized from the
		fully-digital precoder and MMSE combiner as well as predetermined codebook . The
		analog beamforming codebook used consists of the array steering vectors (the
		left/right singular vectors with phase mapping) of the desired mmWave (Rayleigh)
		channel. In particular, three baselines, called \textbf{Hybrid OMP-WMMSE},
		\textbf{Hybrid OMP-ZF} and \textbf{Hybrid OMP-SLNR}, are adopted according to
		three different fully-digital precoders based on the WMMSE, BD-ZF and SLNR-Max
		criteria, respectively.
		
		\textbf{Hybrid EGT-DFT Two-Stage} \cite{ni2016hybrid}: The EGT based analog precoder
		and DFT based analog combiner harvest the large array gain in the first analog stage,
		and the inter-user interference elimination is left to the second digital  stage.
		
		\textbf{Hybrid ParTx-SDR/Hybrid ParTxRx-SDR} \cite{yu2016alternating}: First the
		Euclidean distance between the partially-connected hybrid precoder and the
		fully-digital WMMSE precoder is minimized  in which the iterative procedure
		between the semidefinite relaxation (SDR) based digital precoder and the PP-based
		analog precoder is performed. Then the MM-based hybrid combiner designs under the
		fully connected and partially-connected receiver structures are performed,
		corresponding to \textbf{Hybrid ParTx-SDR} and \textbf{Hybrid ParTxRx-SDR},
		respectively.
		
		\textbf{Analog-only beamsteering} \cite{alkhateeb2015limited}: Only analog beamforming
		strategies at both ends are considered to align transmit and receive beams of each
		transceiver pair for maximizing array gain. The inter-user interference elimination
		is not involved.
		
		Moreover, the near-optimal fully-digital schemes based on the criteria of WMMSE,
		BD-ZF and SLNR-Max (\textbf{Fully-Digital-WMMSE}, \textbf{Fully-Digital-ZF} and 
		\textbf{Fully-Digital-SLNR}) are  adopted as the corresponding upper-bound 
		benchmarks.
		
		Fig.~\ref{fig20} compares the sum rate performance versus SNR in the mmWave
		channel achieved by the MM-Alt-Opt and Hybrid PP-Two-Stage with those of the
		three benchmarks, using the Fully-Digital-WMMSE as the upper bound. It can be
		seen from Fig.~\ref{fig20} that the sum rate of our MM-Alt-Opt is very close
		to the optimal Fully-Digital-WMMSE, confirming that it is near-optimal.
		Benefited from its iterative nature, the MM-Alt-Opt clearly outperforms the
		Hybrid PP-Two-Stage with one-shot approximation for analog precoder and
		combiner design. Also the Hybrid PP-Two-Stage achieves a similar
		performance to the Hybrid OMP-WMMSE at low SNR region, but slightly better
		performance at high SNR region. More importantly, the Hybrid PP-Two-Stage does
		not require the WMMSE fully-digital solution and has much lower-complexity than
		the Hybrid OMP-WMMSE. Since the inter-user interference elimination is not
		considered in the Analog-only beamsteering, its performance is the worst. In
		addition, when a larger number of scatters is considered, i.e., $L\! =\! 12$,
		the  MM-Alt-Opt still performs almost as good as the Fully-Digital-WMMSE, both
		having slightly higher sum rate compared to the case of $L\! =\! 10$.
			\begin{figure}[!tp]
				\vspace{-10mm}
				\begin{center}
					\includegraphics[width=.36\textwidth]{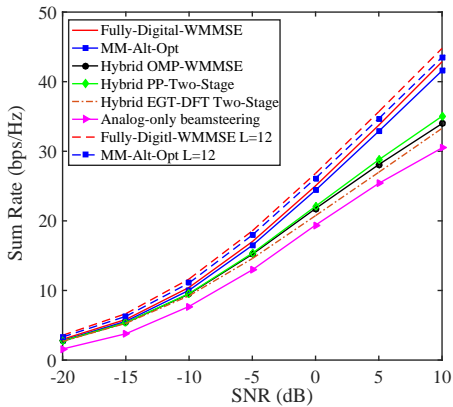}
				\end{center}
				\vspace{-5.5mm}
				\caption{Sum rate performance versus SNR in the mmWave channel achieved by the
					proposed MM-Alt-Opt and Hybrid PP-Two-Stage as well as the benchmarks Hybrid
					OMP-WMMSE, Hybrid EGT-DFT Two-Stage and Analog-only beamsteering, using the
					Fully-Digital-WMMSE as the  upper bound. The sum rates of the MM-Alt-Opt and
					Fully-Digital-WMMSE for the $L\! =\! 12$ scatters are also shown.}
				\label{fig20} 
			\end{figure}
			\begin{figure}[!tp]
				\begin{center}
					\includegraphics[width = .36\textwidth]{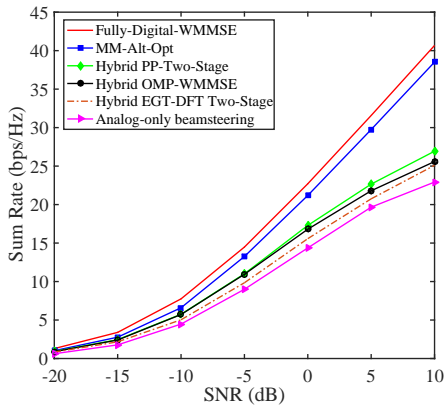}
				\end{center}
				\vspace{-5mm}
				\caption{Sum rate performance versus SNR in the Rayleigh channel achieved by the
					proposed MM-Alt-Opt and Hybrid PP-Two-Stage as well as the benchmarks Hybrid
					OMP-WMMSE, Hybrid EGT-DFT Two-Stage and Analog-only beamsteering, using the
					Fully-Digital-WMMSE as the  upper bound.}
				\label{fig21} 
					\vspace{-2mm}
			\end{figure}
		Next, we carry the same comparison in the Rayleigh scenario, and the results are
		shown in Fig.~\ref{fig21}. Observe that the sum rate gap between the optimal
		Fully-Digital-WMMSE and the MM-Alt-Opt is larger than in the mmWave channel. The
		reason is that the approximation $\bm{F}_{A_k}^{\rm H}\bm{F}_{A_k}\! \!\approx\!\!
		N_t\bm{I}_{N_t^{RF}}$ adopted in the MM-Alt-Opt is less accurate in the Rayleigh
		case.

		\begin{figure}[tp]
			\vspace{-7mm}
			\begin{center}
				\includegraphics[width = .5\textwidth]{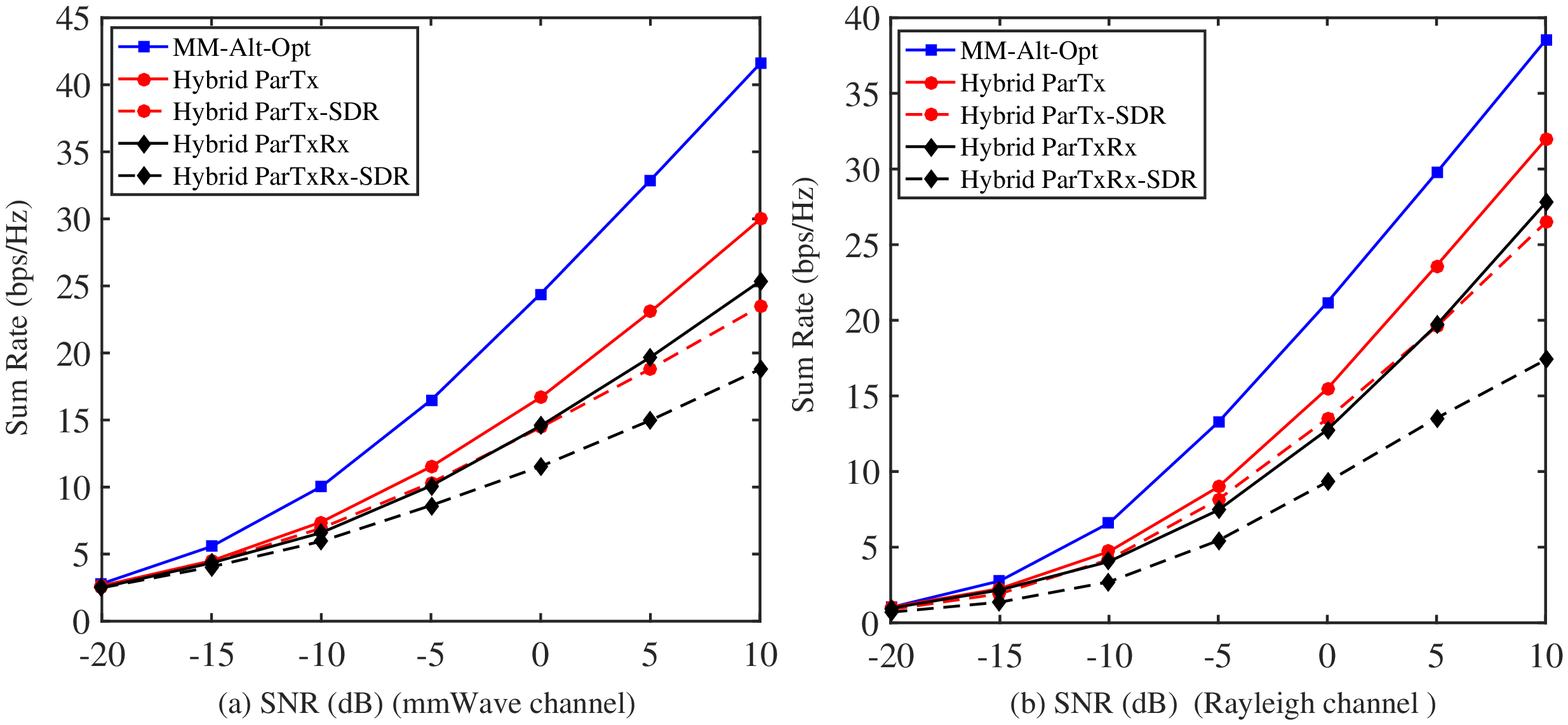}
			\end{center}
			\vspace{-5mm}
			\caption{Sum rate performance versus SNR in (a)~the mmWave channel and (b)~the
				Rayleigh channel, achieved by the proposed MM-Alt-Opt, Hybrid-ParTx and
				Hybrid-ParTxRx as well as the benchmarks Hybrid ParTx-SDR and Hybrid ParTxRx-SDR.}
			\label{fig22} 
		\end{figure}
		
		From Fig.~\ref{fig22}, it can be seen that the MM-Alt-Opt considerably outperforms
		the Hybrid-ParTx in both the mmWave and Rayleigh cases, since  the inter-user
		interference cannot be effectively suppressed by the Hybrid-ParTx with its much
		reduced design freedom in analog precoder. Similarly, the Hybrid-ParTx has better
		sum rate performance than the Hybrid-ParTxRx, since the latter has the further
		much reduced design freedom in analog combiner. Also, observe from Fig.~\ref{fig22}
		that the proposed Hybrid-ParTx outperforms its corresponding benchmark Hybrid
		ParTx-SDR, while the  Hybrid-ParTxRx outperforms its related baseline
		Hybrid ParTxRx-SDR.
		\begin{figure}[tp!]
			\begin{center}
				\includegraphics[width = .36\textwidth]{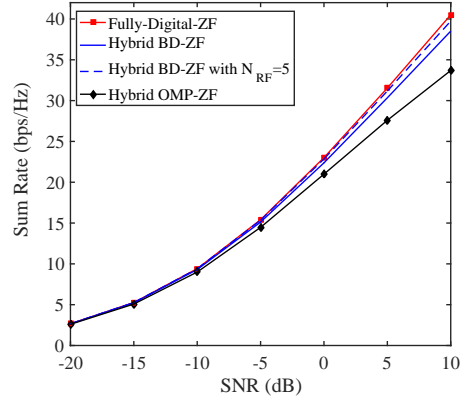}
			\end{center}
			\vspace{-5mm}
			\caption{Sum rate performance versus SNR in the mmWave channel achieved by the
				proposed Hybrid BD-ZF and the benchmark Hybrid OMP-ZF, using the Fully-Digital-ZF
				as the upper bound. The sum rate of the Hybrid BD-ZF with $N_t^{RF}\! =\! N_r^{RF}
				\! = \! N_{RF}\! =\! 5$ is also shown.}
			\label{fig23} 
			\vspace{-2mm}
		\end{figure}
		
		Fig.~\ref{fig23} compares the sum rate performance of the proposed Hybrid BD-ZF and 
		the baseline Hybrid OMP-ZF, using the Fully-Digital-ZF solution as the upper bound.
		Observe that the sum rate of the Hybrid BD-ZF is close to that of the full-digital
		BD-ZF solution, especially when one extra RF, i.e., $N_{RF}=5$, is considered.
		Moreover, the proposed hybrid BD-ZF clearly achieves higher sum rate than the
		hybrid OMP-ZF baseline, because its iterative nature enables the hybrid precoder
		to better approximate the fully-digital solution at the expense of higher
		computational complexity. Furthermore, Fig.~\ref{fig24} shows the sum rates achieved
		by the proposed Hybrid SLNR-Max and the baseline Hybrid-OMP-SLNR versus SNR in the 
		mmWave channel, using the Fully-Digital-SLNR as the upper bound. Clearly,
		Fig.~\ref{fig24} presents similar comparison results among the three schemes to
		Fig.~\ref{fig23}.
		
		Fig.~\ref{fig25} compares the sum rate performance versus the number of RF chains
		$N_{RF}$ in the mmWave channel achieved by the proposed MM-Alt-Opt and hybrid-ParTx
		as well as the optimal Fully-Digital-WMMSE. For the hybrid-ParTx, each of the first
		$N_t^{RF}\!-\!1$ RF chains is connected with $\big\lfloor\frac{N_t}{N_t^{RF}}
		\big\rfloor$ transmit antennas, while the last RF chain is connected with $N_t\! -
		\! (N_t^{RF}\! - \!1)\big\lfloor\frac{N_t}{N_t^{RF}}\big\rfloor$ antennas. It has
		been shown in \cite{sohrabi2016hybrid} that when $N_{RF}\! \ge\! 2N_s$, there exists
		a globally optimal hybrid precoder and combiner design, which perfectly reconstructs
		the fully-digital precoder and combiner, yielding the same sum rate  performance.
		Observe from Fig.~\ref{fig25} that almost identical performance are attained by both
		the Fully-Digital-WMMSE and the MM-Alt-Opt when $N_{RF}\!\ge\! 2N_s\! =\! 8$.
		Obviously, the Hybrid-ParTx cannot perfectly reconstruct the fully-digital design due
		to the reduced design freedom of analog precoder, and the achievable sum rate of the
		Hybrid-ParTx increases with $N_{RF}$ mainly owing to the increased design freedom of
		digital precoder. 
		\begin{figure}[tp!]
			\vspace{-11mm}
			\begin{center}
				\includegraphics[width = .36\textwidth]{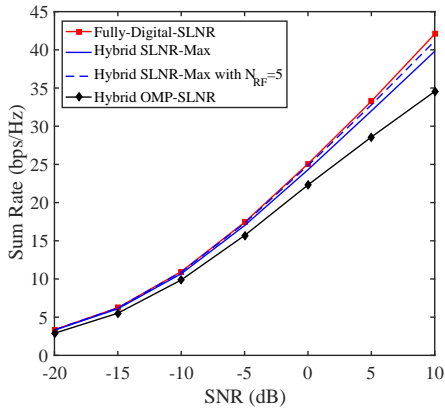}
			\end{center}
			\vspace{-5mm}
			\caption{Sum rate performance versus SNR in the mmWave channel achieved by the
				proposed Hybrid SLNR-Max and the benchmark Hybrid OMP-SLNR, using the
				Fully-Digital-SLNR as the upper bound. The sum rate of the Hybrid SLNR-Max
				with $N_t^{RF}\! =\! N_r^{RF} \! = \! N_{RF}\! =\! 5$ is also shown.}
			\label{fig24} 
		\end{figure}
		\begin{figure}[tp!]
			\begin{center}
				\includegraphics[width =0.36\textwidth]{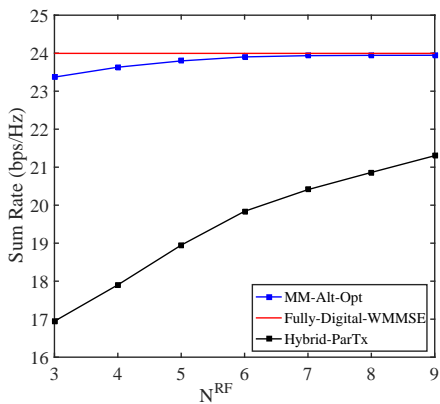}
			\end{center}
			\vspace{-5mm}
			\caption{Sum rate performance versus number of transmit/receive RF chains $N_t^{RF}
				\! =\! N_r^{RF}\! =\! N^{RF}$ in the mmWave channel achieved by the proposed
				MM-Alt-Opt and Hybrid-ParTx, in comparison with the Fully-Digital-WMMSE, given
				$\text{SNR}\! =\! 0$\,dB.}
			\label{fig25} 
			\vspace{-2mm}
		\end{figure}
		\begin{figure}[t]
			\vspace{-9mm}
			\begin{center}
				\includegraphics[width = .52\textwidth]{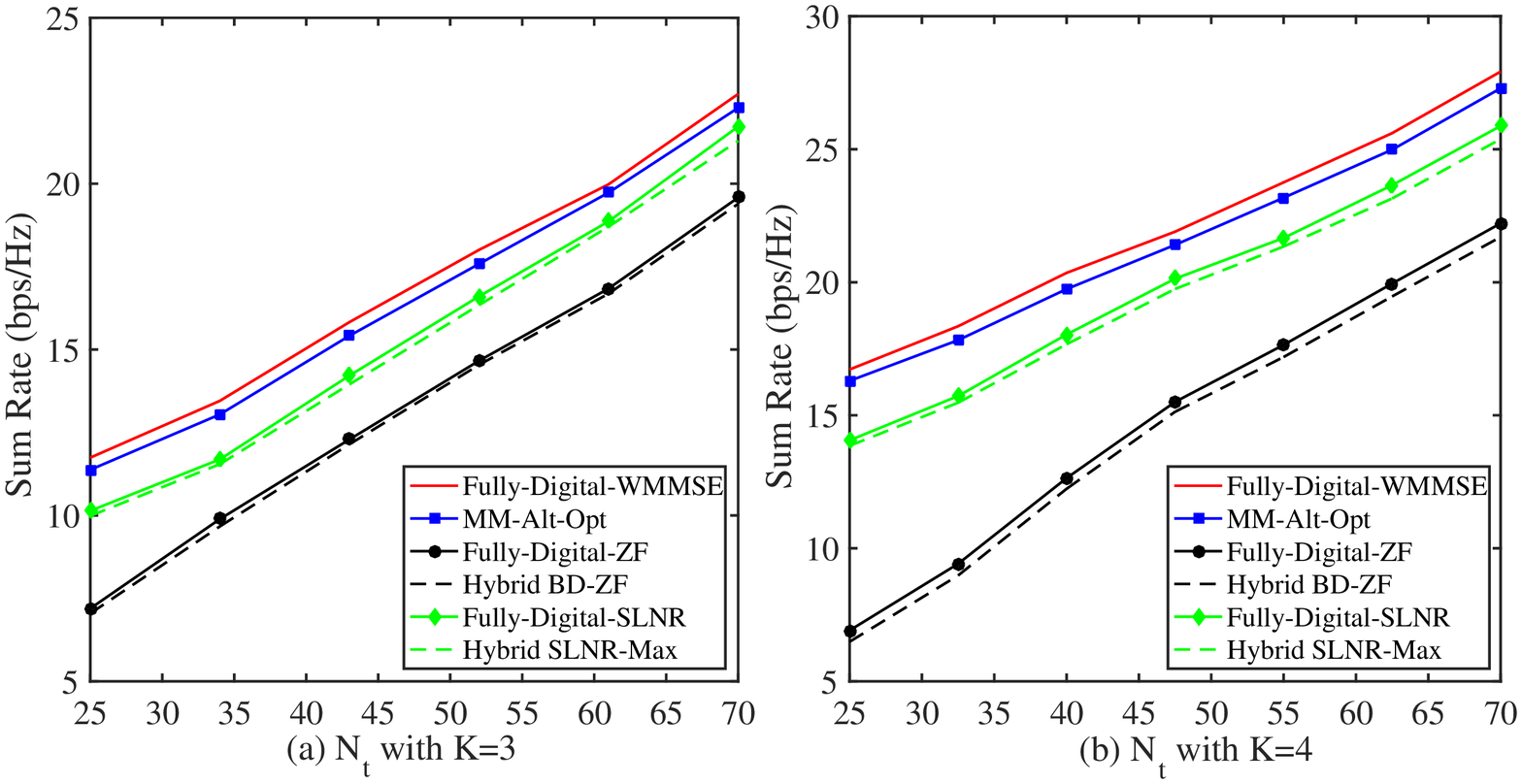}
			\end{center}
			\vspace{-5mm}
			\caption{Sum rate performance versus the number of transmit antennas $N_t$ in the
				mmWave channel achieved by the proposed MM-Alt-Opt, Hybrid BD-ZF and Hybrid
				SLNR-Max, in comparison with the corresponding optimal Fully-Digital-WMMSE,
				Fully-Digital-ZF and Fully-Digital-SLNR: (a)~$K\! =\! 3$ users, and (b)~$K\! =\! 4$
				users.}
			\label{fig26} 
			\vspace{-2mm}
		\end{figure}
		
		Finally, Fig.~\ref{fig26} depicts the sum rates as functions of the number of
		transmit antennas $N_t$ in the mmWave channel achieved by the proposed MM-Alt-Opt,
		Hybrid BD-ZF and Hybrid SLNR-Max, in comparison with their corresponding optimal
		Fully-Digital-WMMSE, Fully-Digital-ZF and Fully-Digital-SLNR designs, respectively.
		The results of Fig.~\ref{fig26} confirm that the MM-Alt-Opt, Hybrid BD-ZF and
		Hybrid SLNR-Max designs are all closed to their respective  fully-digital
		solutions. Obviously, the achievable sum rates of all the schemes increase with
		$N_t$ owing to the increased spatial degrees of freedom. Also, the MM-Alt-Opt
		outperforms the Hybrid SLNR-Max, while the Hybrid SLNR-Max has better sum rate
		than the Hybrid BD-ZF.
		\vspace{-3mm}
		\section{Conclusions}\label{S8}
		
		This paper has investigated various hybrid transceiver designs for sum rate
		maximization in both the mmWave and Rayleigh $K$-user MIMO interference channels.
		First,  bypassing the near-optimal WMMSE fully-digital solution, we have jointly
		designed hybrid precoder and combiner in an alternating manner, in which the
		MM method is used to design the analog precoder and combiner. Moreover, a
		PP-based two-stage scheme has been proposed to decouple the design of analog
		and digital precoder (combiner), leading to lower complexity. Second, with the
		aid of the easy-to-implement fully-digital precoder, the low-complexity BD-ZF and
		SLNR-Max hybrid schemes have been studied, which focus on approximating the hybrid
		precoders to the fully-digital solutions derived according to the BD-ZF and SLNR
		criteria, respectively. Third, the partially-connected transceiver structure has
		been considered to reduce the system hardware cost and complexity, to which the
		MM-based alternating optimization is applicable. Numerical results have
		demonstrated the effectiveness of all our proposed hybrid transceiver designs, and
		they have shown that the sum rate performance of all our proposed hybrid designs
		are close or superior to those of the existing benchmarks. Our future research
		will study all the proposed hybrid  designs implemented with finite resolution
		phase shifters and/or with limited channel feedback.
		\vspace{-3mm}
		\appendix
		
		\subsection{Proof of {\bf Proposition~\ref{p1}}}\label{APa}
		
		\begin{proof}
			The detailed proof of the equality \eqref{Ho1} has been given in \cite{wu2018hybrid}.
			We directly use it to validate the equality \eqref{Ho2}. Firstly, recalling
			$\bm{H}_{i,k}\! =\! \bm{U}_{i,k}\bm{\Lambda}_{i,k}\bm{V}_{i,k}^{\rm H}$, $\forall i,k$,
			we have\vspace{-1mm}
			\begin{align}\label{Ho3} 
			\lim_{N_{t_k}\!\to\! +\infty}\!\! \bm{U}_{i,k}\widetilde{\bm{\Lambda}}_{i,k}\bm{V}_{i,k}^{\rm H}
			\bm{V}_{k,k}\widetilde{\bm{\Lambda}}_{k,k}\bm{U}_{k,k}^{\rm H} \!=\! \bm{0}_{N_{r_i}\times N_{r_k}},
			\forall  i\! \neq \!k ,
			\end{align}
			where $\widetilde{\bm{\Lambda}}_{i,k}\!\! =\!\! \frac{1}{\sqrt{N_{t_k}}}\bm{\Lambda}_{i,k}$.
			Since $\bm{U}_{i,k}$, $\forall i,k$, are unitary, we have
			\begin{align}\label{Ho4} 
			\lim_{N_{t_k} \to +\infty} \widetilde{\bm{\Lambda}}_{i,k}\bm{V}_{i,k}^{\rm H}\bm{V}_{k,k}
			\widetilde{\bm{\Lambda}}_{k,k} =& \bm{0}_{N_{r_i}\times N_{r_k}}, ~ \forall i \neq k .
			\end{align}
			Let $a_{p,q}$ be the $(p,q)$th element of the matrix $\bm{V}_{i,k}^{\rm H}
			(1:\text{rank}(\bm{H}_{i,k}),:) \bm{V}_{k,k}(:,1:\!\text{rank}(\bm{H}_{k,k}))$ with
			$p\! =\! 1\cdots\text{rank}(\bm{H}_{i,k})$ and $q\!=\!1\cdots\text{rank}(\bm{H}_{k,k}\!)$.
			Then, the $(p,q)$th element of the matrix $\widetilde{\bm{\Lambda}}_{i,k}{\bm{V}}_{i,k}^{\rm H}
			\bm{V}_{k,k}\widetilde{\bm{\Lambda}}_{k,k}$ can be expressed as $a_{p,q}
			\big[\widetilde{\bm{\Lambda}}_{i,k}\big]_{p,p}\big[\widetilde{\bm{\Lambda}}_{k,k}\big]_{q,q}$.
			Since the singular values $\big[\widetilde{\bm{\Lambda}}_{i,k}\big]_{l,l}$, $\forall i,k
			\! =\! 1,\cdots, K$, are nonzero when $l\! \le\! \text{rank}(\bm{H}_{i,k})$, we readily
			conclude that the equality \eqref{Ho4} holds if and only if $a_{p,q}\! =\! 0$, $\forall p,q$,
			which leads to \eqref{Ho2}. This completes the proof.
		\end{proof}
		\vspace{-5mm}
		\subsection{Proof of {\bf Proposition~\ref{Le2}}}\label{APb}
		\begin{proof}
			Recalling the mmWave channel model \eqref{eq5},  we have
			\begin{align}\label{eqL1} 
			\bm{H}_{k,i} =& \sqrt{\frac{N_{r_k} N_{t_i}}{L_{k,i}}} \bm{A}_r^k\bm{\Lambda}_{k,i}
			\big(\bm{A}_t^i\big)^{\rm H}, ~ \forall k, i=1,\cdots, K ,
			\end{align}
			where $\bm{A}_r^k\! =\! \big[ \bm{a}_r(\theta_k^1),\cdots ,\bm{a}_r(\theta_k^{L_{k,i}})\big]
			\! \in\! \mathbb{C}^{N_{r_k}\times L_{k,i}}$, $\bm{A}_t^i\! =\! \big[\bm{a}_t(\psi_i^1),\cdots ,
			\bm{a}_t(\psi_i^{L_{k,i}})\big]\! \in\! \mathbb{C}^{N_{t_i}\times L_{k,i}}$ and
			$\bm{\Lambda}_{k,i}\! =\! \text{diag}\big[\alpha_k^1,\cdots ,\alpha_k^{ L_{k,i}}\big]$. Note
			that $N_{s_k}\! \le\! N_{t_k}^{RF}\! \le\! L_{k,k}$ and $L_{k,i}=\text{rank}(\bm{H}_{k,i})$,$\forall i,k,$ are implied. Referring
			to \cite{el2014spatially}, when $N_{t_k}\to +\infty$, the array steering vectors
			$\bm{a}_t^{\rm H}(\psi_i^l)$, $\forall l$, are linearly independent and asymptotically
			orthogonal with probability one, i.e., $\lim\nolimits_{N_{t_k}\to+\infty}
			\bm{a}_t^{\rm H}\big(\psi_i^{l_1}\big)\bm{a}_t\big(\psi_i^{l_2}\big)\! =\! 0$,
			$\forall l_1\! \neq\! l_2$, and $\lim\nolimits_{N_{t_k} \to +\infty}\big(\bm{A}_t^k\big)^{\rm H}
			\bm{A}_t^k\! =\! \bm{I}_{N_{t_k}}$, $\forall k$, which implies that in large-scale mmWave MIMO
			regime, the array response matrix $\bm{A}_t^k$ { can be approximated to} the right singular matrix of
			$\bm{H}_{i,k}$. Furthermore, by recalling \eqref{ZF1} and  exploiting the equality
			\eqref{Ho2}, the fully-digital BD-ZF precoder $\bm{F}_k^{\rm{ZF}}$ can be re-expressed as\vspace{-1mm}
			\begin{align}\label{eq66}
			\lim_{N_{t_k}\! \to +\infty} \!\bm{F}_k^{\rm{ZF}}
			& \!=\! \bm{V}_{k,k}(:,1:L_{k,k})
			\sqrt{\overline{\bm{\Lambda}}_k}\!=\! \bm{A}_t^k \sqrt{\overline{\bm{\Lambda}}_k} ,
			\end{align}
			where {$\overline{\bm{\Lambda}}_{k}\! =\! \text{BLKdiag}\big[ \bm{\Lambda}_k, 
				\bm{0}_{L_{k,k}-N_{s_k},L_{k,k}-N_{s_k}}\big]$} and $\bm{\Lambda}_k$ is determined
			by solving the problem \eqref{ZF2}. Obviously, the matrix $\bm{A}_t^k$ with unit-modulus
			elements can be realized by RF phase shifters, so that the proposed iterative-PP
			analog precoder $\bm{F}_{A_k}$ in \eqref{ZFp2} is easily obtained as $\bm{F}_{A_k}^{\infty}
			\! =\! \bm{A}_t^k(:,1:N_{t_k}^{RF})$ when $N_{t_k}\! \to\! +\infty$. Correspondingly, the
			optimal digital precoder is readily derived as $\bm{F}_{D_k}^{\infty}\! =\! \big[
			\sqrt{\bm{\Lambda}}_k ~ \bm{0}_{N_{s_k}\times (N_{t_k}^{RF}\!-N_{s_k})}\big]^{\rm H}$.
			Using the above hybrid precoder design of the $k$th transmitter, the resultant interference
			at the $i$th receiver, where $i\! \neq\! k$, satisfies  \vspace{-1mm}
			\begin{align}\label{eq67}
			&\lim_{N_{t_k}\to +\infty}\!\!  \bm{H}_{i,k}\bm{F}_{A_k}\bm{F}_{D_k} =
			\! \lim_{N_{t_k}\to +\infty}\!\! \bm{H}_{i,k}\bm{F}_{A_k}^{\infty}\bm{F}_{D_k}^{\infty}
			\nonumber\\
			&\!= \! \!\lim_{N_{t_k}\to \!+\!\infty} \!\!\bm{U}_{i,k}\bm{\Lambda}_{i,k} \bm{V}_{i,k}^{\rm H}(1:{L}_{i,k},:)
			\bm{A}_t^k(:,1:N_{s_k}) \sqrt{\bm{\Lambda}}_k,  \\
			& \!=\!\!\lim_{N_{t_k}\to +\infty}\!\!\!\! \bm{U}_{i,k}\bm{\Lambda}_{i,k}\bm{V}_{i,k}^{\rm H}(1:{L}_{i,k},:)
			\bm{V}_{k,k}(:,1:N_{s_k}) \sqrt{\bm{\Lambda}}_k \!= \!\bm{0}.\nonumber
			\end{align} where the last equality holds by recalling  \eqref{Ho2}. This completes the proof.
		\end{proof}
		\vspace{-4mm}
		\bibliographystyle{IEEEtran}

\begin{thebibliography}{10}
			\bibitem{lu2014overview} 
			L.~Lu, \emph{et~al.}, ``An overview of massive MIMO: Benefits and challenges,''
			\emph{IEEE J. Sel. Topics Signal Process.}, vol.~8, no.~5, pp.~742--758,
			Oct. 2014.
			\bibitem{rusek2012scaling} 
			F.~{Rusek}, \emph{et~al.}, ``Scaling up MIMO: Opportunities and challenges with
			very large arrays,'' \emph{IEEE Signal Process. Mag.}, vol.~30, no.~1, pp.~40--60,
			Jan. 2013.
			
			\bibitem{roh2014millimeter} 
			W.~Roh, \emph{et~al.}, ``Millimeter-wave beamforming as an enabling technology
			for 5G cellular communications: Theoretical feasibility and prototype results,''
			\emph{IEEE Commun. Mag.}, vol.~52, no.~2, pp.~106--113, Feb. 2014.
			\bibitem{hoydis2013massive} 
			J.~Hoydis, S.~Ten~Brink, and M.~Debbah, ``Massive MIMO in the UL/DL of cellular
			networks: How many antennas do we need?'' \emph{IEEE J. Sel. Areas Commun.},
			vol.~31, no.~2, pp.~160--171, Feb. 2013.
			
			\bibitem{molisch2017hybrid} 
			A.~F. Molisch, \emph{et~al.}, ``Hybrid beamforming for massive {MIMO}: A survey,''
			\emph{IEEE Commun. Mag.}, vol.~55, no.~9, pp.~134--141, Sep. 2017.
			\bibitem{el2014spatially} 
			O.~El~Ayach, \emph{et~al.}, ``Spatially sparse precoding in millimeter wave MIMO
			systems,'' \emph{IEEE Trans. Wireless Commun.}, vol.~13, no.~3, pp.~1499--1513,
			Mar. 2014.
			
			\bibitem{rebeiz2002rf} 
			G.~M.~Rebeiz, G.~L.~Tan, and J.~S.~Hayden, ``RF MEMS phase shifters: Design and
			applications,'' \emph{IEEE Microw. Mag.}, vol.~3, no.~2, pp.~72--81, Jun. 2002.
			
			\bibitem{mendez2016hybrid} 
			R.~M\'endez~Rial, \emph{et~al.}, ``Hybrid MIMO architectures for millimeter wave
			communications: Phase shifters or switches?'' \emph{IEEE Access}, vol.~4,
			pp.~247--267, Jan. 2016.
			
			\bibitem{zeng2014electromagnetic} 
			Y.~Zeng, R.~Zhang, and Z.~N.~Chen, ``Electromagnetic lens-focusing antenna
			enabled massive MIMO: Performance improvement and cost reduction,'' \emph{IEEE J.
				Sel. Areas Commun.}, vol.~32, no.~6, pp.~1194--1206, Jun. 2014.
			
			\bibitem{liang2014low} 
			L.~Liang, W.~Xu, and X.~Dong, ``Low-complexity hybrid precoding in massive
			multiuser MIMO systems,'' \emph{IEEE Wireless Commun. Lett.}, vol.~3, no.~6,
			pp.~653--656, Dec. 2014.
			\bibitem{sohrabi2016hybrid} 
			F.~Sohrabi and W.~Yu, ``Hybrid digital and analog beamforming design for
			large-scale antenna arrays,'' \emph{IEEE J. Sel. Topics Signal Process.}, vol.~10,
			no.~3, pp.~501--513, Apr. 2016.
			\bibitem{sohrabi2017hybrid} 
			F.~Sohrabi and W.~Yu, ``Hybrid analog and digital beamforming for mmwave OFDM
			large-scale antenna arrays,'' \emph{IEEE J. Sel. Areas Commun.}, vol.~35, no.~7,
			pp.~1432--1443, Jul. 2017.
			\bibitem{singh2015feasibility} 
			J.~Singh and S.~Ramakrishna, ``On the feasibility of codebook-based beamforming
			in millimeter wave systems with multiple antenna arrays,'' \emph{IEEE Trans.
				Wireless Commun.}, vol.~14, no.~5, pp.~2670--2683, May 2015.
			\bibitem{Liu_etal2018} 
			W.~Liu, \emph{et al.}, ``Partially-activated conjugate beamforming for LoS massive
			MIMO communications,'' \emph{IEEE Access}, vol.~6, pp.~56504--56513, Oct. 2018. 
			\bibitem{ni2016hybrid} 
			W.~Ni and X.~Dong, ``Hybrid block diagonalization for massive multiuser MIMO
			systems,'' \emph{IEEE Trans. Commun.}, vol.~64, no.~1, pp. 201--211, Jan. 2016.
			\bibitem{Zhou_etal2018} 
			Z.~Zhou, N.~Ge, Z.~Wang, and S.~Chen, ``Hardware-efficient hybrid precoding for
			millimeter wave systems with multi-feed reflectarrays,'' \emph{IEEE Access},
			vol.~6, pp.~6795--6806, Mar. 2018.
			\bibitem{Xing2019hybrid} 
			C.~Xing, \emph{et al.}, ``A framework on hybrid MIMO transceiver design based
			on matrix-monotonic optimization,'' \emph{IEEE Trans. Signal Process.}, vol.~67,
			no.~13, pp.~3531--3546, Jul. 2019.
			
			\bibitem{yu2016alternating} 
			X.~Yu, J.~C.~Shen, J.~Zhang, and K.~B.~Letaief, ``Alternating minimization
			algorithms for hybrid precoding in millimeter wave MIMO systems,'' \emph{IEEE J.
				Sel. Topics Signal Process.}, vol.~10, no.~3, pp.~485--500, Apr. 2016.
			\bibitem{ni2017near} 
			W.~Ni, X.~Dong, and W.~S. Lu, ``Near-optimal hybrid processing for massive MIMO
			systems via matrix decomposition,'' \emph{IEEE Trans. Signal Process.}, vol.~65,
			no.~15, pp.~3922--3933, Aug. 2017.
			\bibitem{chen2015iterative} 
			C.~E. Chen, ``An iterative hybrid transceiver design algorithm for millimeter
			wave MIMO systems,'' \emph{IEEE Wireless Commun. Lett.}, vol.~4, no.~3, pp.~285--288,
			Jun. 2015.
			
			\bibitem{kim2015mse} 
			M.~Kim and Y.~H. Lee, ``MSE-based hybrid RF/baseband processing for millimeter-wave
			communication systems in MIMO interference channels,'' \emph{IEEE Trans. Veh.
				Techno.}, vol.~64, no.~6, pp.~2714--2720, Jun. 2015.
			\bibitem{nguyen2017hybrid} 
			D.~H.~Nguyen, L.~B.~Le, T.~Le-Ngoc, and R.~W.~Heath, ``Hybrid MMSE precoding
			and combining designs for mmwave multiuser systems,'' \emph{IEEE Access}, vol.~5,
			pp.~19167--19181, Sep. 2017.
			\bibitem{rajashekar2017iterative} 
			R.~Rajashekar and L.~Hanzo, ``Iterative matrix decomposition aided block
			diagonalization for mm-wave multiuser MIMO systems,'' \emph{IEEE Trans.
				Wireless Commun.}, vol.~16, no.~3, pp. 1372--1384, Mar. 2017.
			
			\bibitem{alkhateeb2015limited} 
			A.~Alkhateeb, G.~Leus, and R.~W.~Heath, ``Limited feedback hybrid precoding for
			multi-user millimeter wave systems,'' \emph{IEEE Trans. Wireless Commun.},
			vol.~14, no.~11, pp.~6481--6494, Nov. 2015.
			\bibitem{wu2018hybrid} 
			X.~Wu, D.~Liu, and F.~Yin, ``Hybrid beamforming for multi-user massive MIMO
			systems,'' \emph{IEEE Trans. Commun.}, vol.~66, no.~9, pp. 3879--3891,
			Sep. 2018.
			
			\bibitem{spencer2004zero} 
			Q.~H.~Spencer, A.~L.~Swindlehurst, and M.~Haardt, ``Zero-forcing methods for
			downlink spatial multiplexing in multiuser MIMO channels,'' \emph{IEEE Trans.
				Signal Process.}, vol.~52, no.~2, pp.~461--471, Feb. 2004.
			
			\bibitem{sadek2007leakage} 
			M.~Sadek, A.~Tarighat, and A.~H.~Sayed, ``A leakage-based precoding scheme
			for downlink multi-user MIMO channels,'' \emph{IEEE Trans. Wireless Commun.},
			vol.~6, no.~5, pp.~1711--1721, May 2007.
			
			\bibitem{cheng2010new} 
			P.~Cheng, M.~Tao, and W.~Zhang, ``A new SLNR-based linear precoding for
			downlink multi-user multi-stream MIMO systems,'' \emph{IEEE Commun. Lett.},
			vol.~14, no.~11, pp.~1008--1010, Nov. 2010.
			
			\bibitem{liu2014phase} 
			A.~Liu and V.~Lau, ``Phase only RF precoding for massive MIMO systems with
			limited RF chains,'' \emph{IEEE Trans. Signal Process.}, vol.~62, no.~17,
			pp.~4505--4515, Sep. 2014.
			\bibitem{liu2015two} 
			A.~Liu and V.~K.~Lau, ``Two-stage subspace constrained precoding in massive
			MIMO cellular systems,'' \emph{IEEE Trans. Wireless Commun.}, vol.~14, no.~6,
			pp.~3271--3279, Jun. 2015.
			
			\bibitem{sun2017majorization} 
			Y.~Sun, P.~Babu, and D.~P.~Palomar, ``Majorization-minimization algorithms in
			signal processing, communications, and machine learning,'' \emph{IEEE Trans.
				Signal Process.}, vol.~65, no.~3, pp.~794--816, Feb. 2017.
			\bibitem{wu2018transmit} 
			L.~Wu, P.~Babu, and D.~P.~Palomar, ``Transmit waveform/receive filter design
			for MIMO radar with multiple waveform constraints,'' \emph{IEEE Trans. Signal
				Process.}, vol.~66, no.~6, pp.~1526--1540, Mar. 2018.
			\bibitem{Boulingand} 
		J. Pang,``Partially B-regular optimization and equilibrium problems,'' \emph{
		Math. Oper. Res.}, vol.~32, no.~3, pp.~687--699, 2007.
			\bibitem{shi2011iteratively} 
			Q.~Shi, M.~Razaviyayn, Z.~Q.~Luo, and C.~He, ``An iteratively weighted MMSE
			approach to distributed sum-utility maximization for a MIMO interfering broadcast
			channel,'' \emph{IEEE Trans. Signal Process.}, vol.~59, no.~9, pp.~4331--4340,
			Sep. 2011.
			
			\bibitem{wallace2002modeling} 
			J.~W.~Wallace and M.~A.~Jensen, ``Modeling the indoor MIMO wireless channel,''
			\emph{IEEE Trans. Antennas Propag.}, vol.~50, no.~5, pp.~591--599, May 2002.
			
			\bibitem{booktypical} 
			A.~W. Marshall, I.~Olkin, and B.~C. Arnold, \emph{Inequalities: Theory of
				Majorization and Its Applications}. Springer-Verlag: New York, 2011.
			
			\bibitem{jacobson2007expanded} 
			M.~W.~Jacobson and J.~A.~Fessler, ``An expanded theoretical treatment of
			iteration-dependent majorize-minimize algorithms,'' \emph{IEEE Trans. Image
				Process.}, vol.~16, no.~10, pp.~2411--2422, Oct. 2007.
			
		\end{thebibliography}
		
	\end{document}